\documentclass[11pt]{article}
\usepackage[T1]{fontenc}
\usepackage[utf8]{inputenc}
\usepackage{color}
\usepackage{array}
\usepackage{booktabs}
\usepackage{url}
\usepackage{multirow}
\usepackage{amsmath}
\usepackage{amssymb}
\usepackage{graphicx}
\usepackage{subfig}
\usepackage[authoryear]{natbib}
\usepackage[unicode=true,
 bookmarks=false,
 breaklinks=false,pdfborder={0 0 1},
 colorlinks=false]
 {hyperref}

\makeatletter

\providecommand{\tabularnewline}{\\}

\@ifundefined{date}{}{\date{}}

\usepackage{latexsym}
\usepackage{multirow}\usepackage{bm}
\usepackage{url}
\usepackage{tabularx}
\usepackage{booktabs}
\usepackage[table]{xcolor}
\usepackage{setspace}

\usepackage{enumerate}

\newenvironment{proof}{\noindent\textit{Proof:}}{\hfill$\square$}

\newcommand*{\indep}{%
	\rotatebox[origin=c]{90}{$\models$}
}

\usepackage{dcolumn}
\newcolumntype{.}{D{.}{.}{-1}}
\newcolumntype{d}[1]{D{.}{.}{#1}}
\usepackage{theorem}
\theoremstyle{definition}
\newtheorem{assumption}{Assumption}\newtheorem{example}{Example}\newtheorem{theorem}{Theorem}\newtheorem{remark}{Remark}\newtheorem{thm}{Theorem}\newtheorem{lemma}{Lemma}

\usepackage{rotating}

\usepackage{arydshln}

\usepackage[compact]{titlesec}

\allowdisplaybreaks

\newcommand{\spacingset}[1]{\renewcommand{\baselinestretch}%
{#1}\small\normalsize}

\newcommand{\pr}{\mathbb{P}}

\newcommand{\E}{\mathbb{E}}

\newcommand{\V}{\mathbb{V}}

\newcommand{\T}{\text{T}}

\newcommand{\I}{\mathbf{I}}

\makeatother

\begin{document}
\title{\textbf{Robust analyses for longitudinal clinical trials with missing
and non-normal continuous outcomes}}

\author{Siyi Liu$^1$, Yilong Zhang$^2$, Gregory T Golm$^{2}$, Guanghan (Frank) Liu$^{2,3}$, Shu Yang$^1$} 

\date{\vspace{-5ex}}

\maketitle
\begin{center}
$^{1}$Department of Statistics, North Carolina State University, Raleigh, NC, USA \\
$^{2}$Merck \& Co., Inc., Kenilworth, NJ, USA \\
$^{3}$Posthumous
\end{center}



\spacingset{1.5} 
\begin{abstract}
Missing data is unavoidable in longitudinal clinical trials, and outcomes
are not always normally distributed. In the presence of outliers or
heavy-tailed distributions, the conventional multiple imputation with
the mixed model with repeated measures analysis of the average treatment
effect (ATE) based on the multivariate normal assumption may produce
bias and power loss. Control-based imputation (CBI) is an approach
for evaluating the treatment effect under the assumption that participants
in both the test and control groups with missing outcome data have
a similar outcome profile as those with an identical history in the
control group. We develop a general robust framework to handle non-normal
outcomes under CBI without imposing any parametric modeling assumptions.
Under the proposed framework, sequential weighted robust regressions
are applied to protect the constructed imputation model against non-normality
in both the covariates and the response variables. Accompanied by
the subsequent mean imputation and robust model analysis, the resulting
ATE estimator has good theoretical properties in terms of consistency
and asymptotic normality. Moreover, our proposed method guarantees
the analysis model robustness of the ATE estimation, in the sense
that its asymptotic results remain intact even when the analysis model
is misspecified. The superiority of the proposed robust method is
demonstrated by comprehensive simulation studies and an AIDS clinical
trial data application.

\noindent \textbf{keywords:} Longitudinal clinical trial; missing
data; multiple imputation; robust regression; sensitivity analysis.
\end{abstract}
\newpage{}

\section{Introduction}

\subsection{Missing data in clinical trials}

Analysis of longitudinal clinical trials often presents difficulties
as inevitably some participants do not complete the study, thereby
creating missing outcome data. Additionally, some outcome data among
participants who complete the study may not be of interest on account
of intercurrent events such as initiation of rescue therapy prior
to the analysis time point. With the primary interest focusing on
evaluating the treatment effect in longitudinal clinical trials, the
approach to handling missingness plays an essential role and has gained
substantial attention from the US Food and Drug Administration (FDA)
and National Research Council \citep{little2012prevention}. The ICH
E9(R1) addendum provides a detailed framework of defining estimands
to target the major clinical question in a population-level summary
with the consideration of intercurrent events that may cause additional
missingness \citep{international2019addendum}. 

The missing at random (MAR; \citealp{rubin1976inference}) mechanism
is often invoked in analyses that seek to evaluate the treatment efficacy.
However, MAR is unverifiable and may not be practical in some clinical
trials. Further, if the response at the primary time point is of interest
regardless of whether participants have complied with the test or comparator
treatments through the primary time point (corresponding to a \textquoteleft treatment
policy\textquoteright{} intercurrent event strategy), an analysis
based on the MAR assumption would not be appropriate, because such
an analysis would assume that responses in those who drop out would
follow the same trajectory as responses in those who remain in treatment.
A more plausible assumption would be that the treatment effect may
quickly fade away, leading to a missing not at random (MNAR) assumption
that responses among those who fail to complete treatments in both
treatment groups behave similarly to the responses among those in
the control group with identical historical covariates. Drawn on the
idea of the zero-dose model in \citet{little1996intent}, \citet{carpenter2013analysis}
refer to this scenario as the control-based imputation (CBI). Since
the CBI represents a deviation from MAR, it is widely used in sensitivity
analyses to explore the robustness of the study results against the
untestable MAR assumption (e.g., \citealp{carpenter2013analysis,cro2016reference}).
Furthermore, an increasing number of clinical studies have applied
this approach to primary analyses \citep{tan2021review}. Throughout
the paper, we focus on jump-to-reference (J2R) as one favorable scenario
of the CBI used in the FDA statistical review and evaluation reports
(e.g., \citealp{cr2016tresiba}), which assumes that the missing outcomes
in the treatment group will have the same outcome mean profile as
those with identical historical information in the control group.
Our goal is to assess the average treatment effect (ATE) under J2R.

\subsection{Multiple imputation}

Multiple imputation (MI; \citealp{rubin2004multiple}) followed by
a mixed-model with repeated measures (MMRM) analysis acts as a standard
approach to analyze longitudinal clinical trial data under J2R. The
main idea of MI applied in longitudinal trials is to use MMRM to impute
the missing components and then conduct full-data analysis on each
imputed dataset. The simple implementation and high flexibility of
MI underlie the recommendation of this approach by the FDA and National
Research Council \citep{little2012prevention}. 

However, this approach relies heavily on the parametric modeling assumptions
in the construction of both the imputation and the analysis model,
where a normal distribution is typically assumed. In reality, the
distribution of the outcomes may suffer from extreme outliers or a
heavy tail, which contradicts the normality assumption. A motivating
CD4 count dataset in Section \ref{sec:moti_exmp} further addresses
that a simple transformation such as the log transformation sometimes
cannot fix the non-normality issue \citep{mehrotra2012analysis}.
In the presence of outliers or heavy tails, applying the methods that
rely on the normal distribution may produce bias and power loss. To
tackle the issue in longitudinal clinical trials under MAR, \citet{mogg2007analysis}
and \citet{mehrotra2012analysis} suggest substituting the conventional
analysis of covariance model in the full-data analysis step of MI
with th\textcolor{black}{e rank-based reg}ression \citep{jaeckel1972estimating}
or Huber robust regression \citep{huber1973robust} to down-weight
the impact of non-normal response values. When the missingness mechanism
is MNAR, a gap exists in the extension of the robust method to handle
the MNAR-related scenarios. 

\subsection{Our contribution: a robust framework}

We develop a general robust framework to evaluate the ATE for non-normal
longitudinal outcomes with missingness under the scenario where missing
response data in both the test and reference groups are assumed to
follow the same trajectory as the complete data in the reference group.
We propose applying robust regression in conjunction with mean imputation
to relax the parametric modeling assumption required by MI in both
the imputation and analysis stages. Inspired by the sequential linear
regression model involved in many longitudinal studies, where the
current outcomes are regressed recursively on the historical information
\citep{tang2017efficient}, we replace the least squares (LS) estimator
with the estimator obtained by minimizing the robust loss function
such as the Huber loss, the absolute loss \citep{huber2004robust},
and the $\varepsilon$-insensitive loss \citep{smola2004tutorial},
to mitigate the impact of non-normality in the response variable.
While the robust regression lacks the protection against outliers
in the covariates \citep{chang2018robust}, a weighted sequential
robust regression model is put forward using the idea in \citet{carroll1993robustness}
to down-weight the influential covariates by a robust Mahalanobis
distance. Followed by mean imputation and a robust analysis step,
the estimator from our proposed method has solid theoretical guarantees
in terms of consistency and asymptotic normality. 

\citet{rosenblum2009using} establish a test robustness result for randomized clinical trials with complete data; i.e., for a wide range of analysis models, testing the existence of the non-zero ATE has an asymptotically correct type-1 error even under model misspecification. However, they focus only on the LS model estimators when no ATE exists; and the property remains unclear when the model is estimated via the robust loss function under any arbitrary ATE value. To uncover the ambiguity, we extend the test robustness property to our proposed method for  ATE estimation in the context of missing data. We formally show that the ATE estimator obtained from the various non-LS loss functions, including the Huber loss, the absolute loss, and the $\epsilon$-insensitive loss is analysis model-robust, in the sense that its asymptotic properties remain the same even when the analysis model is incorrectly specified. Although the paper mainly focuses on the J2R scenario, the established method and the desired theoretical properties are extendable to robust estimators under other MNAR-related conditions.

The rest of the paper is organized as follows. Section \ref{sec:moti_exmp}
addresses a real-data example to motivate the demand for the robust
method. Section \ref{sec:setup} introduces notations, assumptions
under J2R, and an overview of the existing methods to handle missingness
along with their drawbacks in the presence of non-normal data. Section
\ref{sec:robust} presents our proposed robust method and its detailed
implementation steps. Section \ref{sec:theo} provides the asymptotic
results of the ATE estimator and discusses the analysis model robustness
property. Section \ref{sec:simu} conducts comprehensive simulation
studies to validate the proposed method. Section \ref{sec:app} returns
to the motivating example to illustrate the performance of the robust
method in practice. Section \ref{sec:conclusion} draws the conclusion.

\section{A motivating application \label{sec:moti_exmp}}

Study 193A conducted by the AIDS Clinical Trial Group compares the
effects of dual or triple combinations of the HIV-1 reverse transcriptase
inhibitors \citep{henry1998randomized}. The data consists of the
longitudinal outcomes of the CD4 count data at baseline and during
the first 40 weeks of follow-up, with the fully-observed baseline
covariates as age and gender. In the trial, the participants are
randomly assigned among the four treatments regarding dual or triple
therapies. We focus on the treatment comparison between arm 1 (zidovudine
alternating monthly with 400 mg didanosine) and arm 2 (zidovudine
plus 400mg of didanosine plus 400mg of nevirapine). As arm 1 involves
fewer combinations of inhibitors than arm 2, we view it as the reference
group. Among individuals in these two arms, we delete the ones with
missing baseline CD4 counts, partition the time into discrete intervals
$(0,12]$, $(12,20]$, $(20,28]$, $(28,36]$ and $(36,40]$, and
create a dataset with a monotone missingness pattern.

Since the original CD4 counts are highly skewed, we conduct a log
transformation to get the transformed CD4 counts as $\log(\text{CD4}+1)$
and use them as the outcomes of interest. Figure \ref{fig:sp_real}
presents the spaghetti plots of the transformed CD4 counts. Although
there are no outstanding outliers, severe missingness is evident in
the data, with only 34 of 320 participants in arm 1 and 46 of 330
participants in arm 2 completing the trial. The high dropout rates
in this data reflect a typical missing data issue in longitudinal
clinical trials, leading to the demand of conducting imputation for
the missing components to prevent the substantial information loss
if we focus on only the complete data. 

\begin{figure}
\caption{Spaghetti plots of the log-transformed CD4 count data separated by
the two treatments. \label{fig:sp_real}}

\centering{}\includegraphics[scale=0.4]{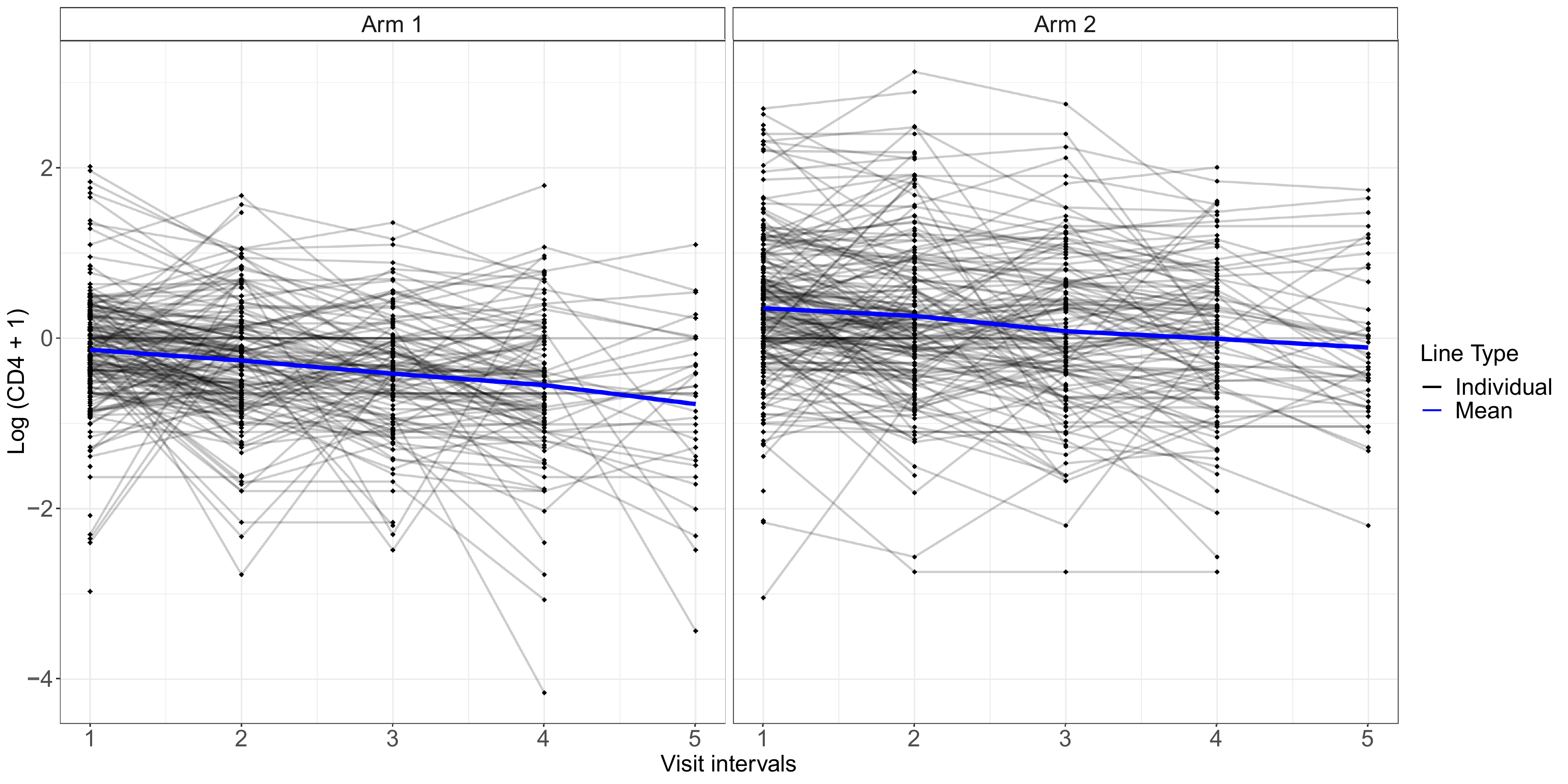}
\end{figure}

We check the normality of the data by fitting sequential linear
regressions on the current outcomes against all historical information
and examining the conditional residuals at each visit point for model
diagnosis. An assessment of the normality of the responses via the
Shapiro-Wilk test and the normal QQ plots are presented in Figure \ref{fig:diagnosis}.
Each normality test indicates a violation of the normal assumption, and the normal QQ plots reveal
that the CD4 counts remain heavy-tailed even after the log transformation.
Under this circumstance, potentially biased and inefficient treatment
effect estimates may occur when applying the conventional MI along
with the MMRM analysis. It motivates the development of a robust method
to assess the treatment effect precisely under non-normality.



\begin{figure}
\caption{Diagnosis of the conditional residuals at each visit.\label{fig:diagnosis}}

\centering{}\includegraphics[scale=0.6]{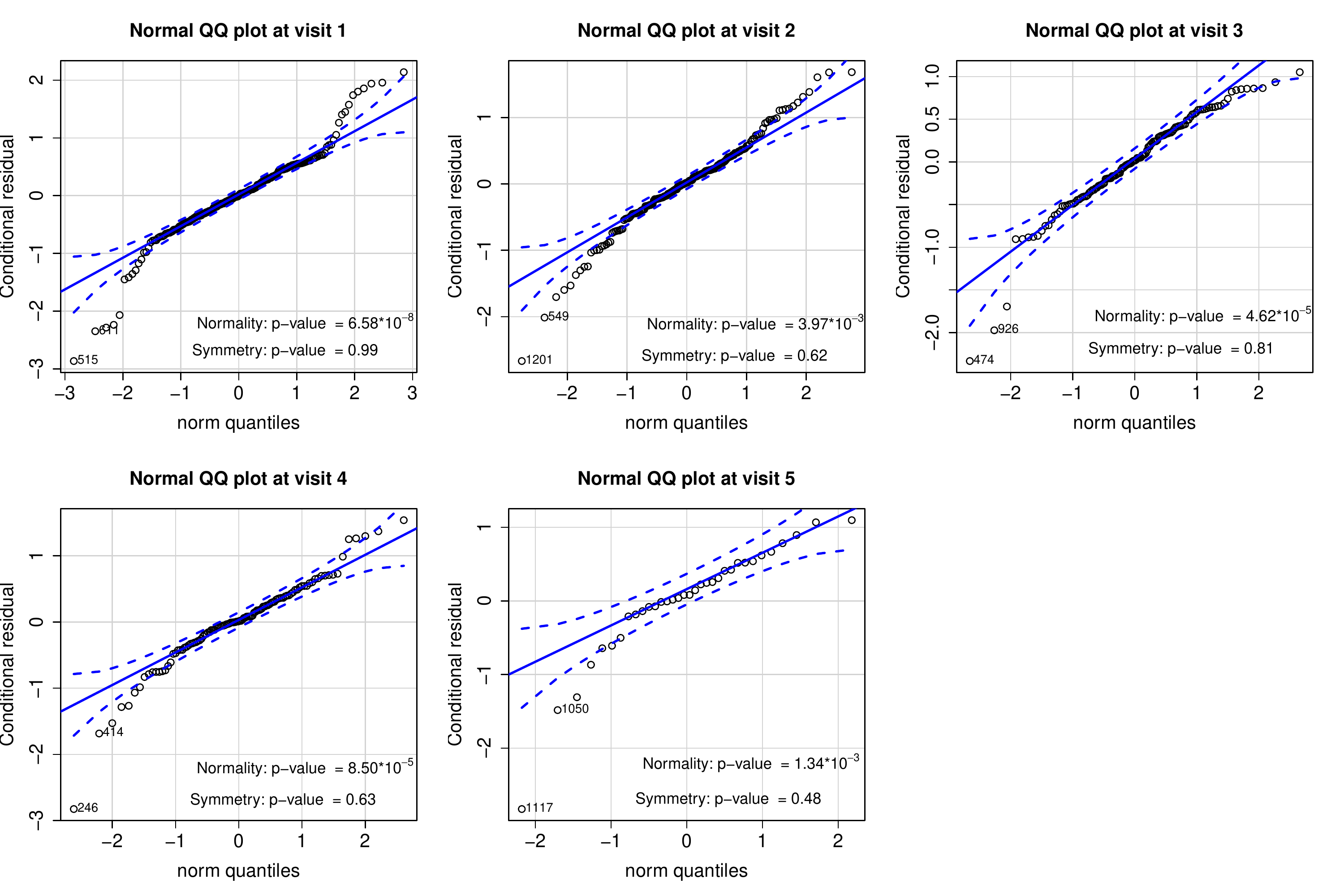}
\end{figure}

\section{Basic setup \label{sec:setup}}

Consider a longitudinal clinical trial with $n$ participants and
$t$ follow-up visits. Let $A_{i}$ be the binary treatment without
loss of generality, $X_{i}$ be the $p$-dimensional fully-observed
baseline covariates including the intercept term with a full-column
rank, $Y_{is}$ be the continuous outcome of interest at visit $s$,
where $i=1,\cdots,n$, and $s=1,\cdots,t$. In longitudinal clinical
trials, participants are randomly assigned to different treatment
groups with non-zero probabilities. When missingness is involved,
denote the observed indicator at visit $s$ as $R_{is}$, where $R_{is}=1$
if $Y_{is}$ is observed and $R_{is}=0$ otherwise. We assume a monotone
missingness pattern throughout the paper, i.e., if the missingness begins
at visit $s$, we have $R_{is'}=1$ for $s'<s$ and $R_{is'}=0$ for
$s'\geq s$. Denote $H_{is}=(X_{i}^{\T},Y_{i1},\cdots,Y_{is})^{\T}$
as the history up to visit $s$, with $H_{i0}=X_{i}$. Since
the outliers in the baseline covariates can be identified and removed
by data inspection before further analysis, throughout we assume
that no outliers exist in the baseline covariates. However, outliers
may exist in the longitudinal outcomes due to data-collection error
in the long period of study.

In most longitudinal clinical trials with continuous outcomes, the
endpoint of interest is the mean difference of the outcomes at the
last visit point between the two treatments. We utilize the pattern-mixture
model (PMM; \citealp{little1993pattern}) framework to express the
ATE as a weighted average over the missing patterns, i.e., $\tau=\mathbb{E}(Y_{it}\mid A_{i}=1)-\mathbb{E}(Y_{it}\mid A_{i}=0)$
where $\E(Y_{it}\mid A=a)=\sum_{s=1}^{t+1}\E(Y_{it}\mid R_{is-1}=1,R_{is}=0,A_{i}=a)\pr(R_{is-1}=1,R_{is}=0\mid A_{i}=a)$
if we let $R_{i0}=1$ and $R_{it+1}=0$ for each individual. The assumed
condition regarding the missing components is formed for the identification
of the pattern-specific expectation $\E(Y_{it}\mid R_{is-1}=1,R_{is}=0,A_{i}=a)$.
We describe one scenario based on the CBI model proposed by \citet{carpenter2013analysis}
for illustration.

\subsection{Jump-to-reference imputation model}

The CBI model \citep{carpenter2013analysis} provides a scenario to
model missingness in longitudinal clinical trials. We focus on one
specific CBI model as J2R, whose plausibility reveals if the investigators
believe that participants who discontinue the treatment have the same
outcome mean performance as the ones in the control group with the
same covariates. The following assumptions illustrate the J2R imputation
model for the ATE identification.

\begin{assumption}[Partial ignorability of missingness]\label{assump:miss}
$R_{is}\indep Y_{is'}\mid(H_{is-1},A_{i}=0)$ for $s'\geq s$.

\end{assumption}

Assumption \ref{assump:miss} characterizes the MNAR missing mechanism
under J2R. The conventional MAR assumption is only required for the
missing data in the control group. We do not impose any missing assumptions
in the treatment group.

\begin{assumption}[J2R outcome mean model]\label{assump:j2r-mean}
For individuals who receive treatment $a$ with historical information
$H_{is-1}$ and drop out at visit $s$, $\E(Y_{it}\mid H_{is-1},R_{is-1}=1,R_{is}=0,A_{i}=a)=\E(Y_{it}\mid H_{is-1},A_{i}=0).$

\end{assumption}

Assumption \ref{assump:j2r-mean} offers a strategy to model the conditional mean
of the missing component under J2R. Given the same historical information,
the outcome mean will ``jump'' to the same conditional mean in the
control group no matter the prior treatment. Combining with Assumption
\ref{assump:miss}, the conditional expectation $\E(Y_{it}\mid H_{is-1},A_{i}=0)=\E\big\{\cdots\E(Y_{it}\mid H_{it-1},R_{it}=1,A_{i}=0)\cdots\mid H_{is-1},R_{is}=1,A_{i}=0\big\}$
is identified through a series of sequential regressions on the current
outcome against the available historical information. 

Throughout the paper, we assume a linear relationship between the
outcomes and the historical covariates in the J2R imputation model
for simplicity.\textcolor{black}{{} Extensions to nonlinear relationships
are manageable, if the sequential regressions of the observed data
are fitted in backward order, i.e., we start from the available data
at the last visit point and use the predicted value as the outcome
to regress on the previous history recursively to construct the imputation
model. The elaboration of the sequential fitting procedure is provided
in Section \ref{sec:supp_seqreg} in the supplementary material. }

\subsection{Overview of the existing methods and the drawbacks \label{subsec:overview}}

MI proposed by \citet{rubin2004multiple} provides a fully parametric
approach to handle missingness under MNAR. Normality is often assumed
for its simplicity and robustness against moderate model misspecification
in the implementation of MI \citep{mehrotra2012analysis}. One common
MI procedure in longitudinal clinical trials under J2R is summarized
in the following steps: 

\begin{enumerate}
\setlength{\itemindent}{1.5em}

\item[\textbf{Step 1}.] For the control group, fit the sequential
regression for the observed data at each visit point against the available
history. Denote the estimated model parameter as $\hat{\theta}_{s-1}$
for $s=1,\cdots,t$. 

\item[\textbf{Step 2}.] Impute missing data sequentially to form
$M$ imputed datasets: For individuals who have missing values at
visit $s$, impute $Y_{is}^{(m)}$ from the conditional distribution
$f(Y_{is}\mid H_{is-1}^{*(m)},A_{i}=0;\hat{\theta}_{s-1})$ estimated
in Step 1, where $H_{is-1}^{*(m)}=(X_{i}^{\T},Y_{i1}^{*(m)},\cdots,Y_{is-1}^{*(m)})^{\T}$
and $Y_{is}^{*(m)}=R_{is}Y_{is}+(1-R_{is})Y_{is}^{(m)}$ for $m=1,\cdots,M$.

\item[\textbf{Step 3}.] For each imputed dataset, perform the complete
data analysis by fitting the imputed outcomes at the last visit point
on a working analysis model. Denote $\hat{\tau}^{(m)}$ as the ATE
estimator of the $m$th imputed dataset.

\item[\textbf{Step 4}.] Combine the estimation results from $M$
imputed datasets and obtain the MI estimator as $\hat{\tau}_{\text{MI}}=M^{-1}\sum_{m=1}^{M}\hat{\tau}^{(m)}$,
with the variance estimator by Rubin's rule as 
\[
\mathbb{\hat{V}}(\hat{\tau}_{\text{MI}})=\frac{1}{M}\sum_{m=1}^{M}\mathbb{\hat{V}}(\hat{\tau}^{(m)})+(1+\frac{1}{M})B_{\text{M}},
\]
where $B_{\text{M}}=(M-1)^{-1}\sum_{m=1}^{M}(\hat{\tau}^{(m)}-\hat{\tau}_{\text{MI}})^{2}$
is the between-imputation variance. 

\end{enumerate}

Traditionally, the imputation model in Step 1 and the analysis model
in Step 3 are obtained from the MMRM analysis,
where we assume an underlying normal distribution for both the observed
and the imputed data. However, as illustrated in the motivating CD4
count dataset in Section \ref{sec:moti_exmp}, normality may be violated,
leading to a biased estimate of the target ATE parameter. With the
consideration of non-normality, \citet{mogg2007analysis} and \citet{mehrotra2012analysis}
modify the analysis model in Step 3 by replacing the LS estimator
with the estimator obtained from the robust loss function.

One drawback of MI is that it is fully parametric. The consistency
of the MI estimator relies heavily on the correct specification of
the imputation distribution, i.e., the conditional distribution given
the observed data, which is often assumed to be normal. When a severe
deviation from the assumed imputation distribution is detected in
the data, the estimation may not be reliable. The possible misspecification
of the imputation distribution also exists in the ``robust'' approaches
proposed by \citet{mogg2007analysis} and \citet{mehrotra2012analysis},
where the imputation model still depends on the normality assumption
as required by MI. Moreover, the MI estimator is not efficient in
general. The inefficiency becomes more serious when it comes to interval
estimation. The variance estimation using Rubin's combining rule may
produce an inconsistent variance estimate even when the imputation
and analysis models are the same correctly specified models \citep{wang1998large,robins2000inference}.
Under the MNAR assumption, the overestimation issue raised from Rubin's
variance estimator is more pronounced (e.g., \citealp{lu2014analytic,liu2016analysis,yang2016note,guan2019unified,yang2020smim,di2022}).
One can resort to the bootstrap variance estimation to obtain a consistent
variance estimator, which however exaggerates the computational cost. 

The unsatisfying performance of MI under non-normality motivates us
to develop a robust approach to accommodate the possible model misspecification
resulting from outliers or heavy-tailed errors without the reliance
on parametric models. In the following sections, a weighted robust
regression model in conjunction with mean imputation is proposed to
overcome the issues in MI. 

\section{Proposed robust method \label{sec:robust}}

We propose a mean imputation procedure based on robust regression
in both the imputation and the analysis models to obtain valid inferences
under J2R when the data suffers from a heavy tail or extreme outliers.
To relax the strong parametric modeling assumption required by MI,
mean imputation is preferred. \citet{mehrotra2012analysis} shed light
on the possibility of incorporating the robust regression in the analysis
step of MI to handle non-normality. Based on this idea, we further
suggest using the sequential robust regression model in the imputation
step to protect against deviations from normality for the observed
data.

Throughout this section, we focus on the robust estimators obtained
from minimizing the robust loss functions such as the Huber loss,
the absolute loss \citep{huber2004robust}, and the $\varepsilon$-insensitive loss \citep{smola2004tutorial} to account
for the impact of outliers or heavy-tailed data. To obtain a valid
mean-type estimator, a symmetric error distribution assumption is
imposed whenever a robust regression is applied.

Motivated by the sequential linear regression model under normality,
where we regress the current outcomes on the historical information
at each visit point to produce the sequential inferences, we develop
a sequential robust regression procedure for the observed data to
obtain valid inferences that are less likely to be influenced by non-normality.
For the longitudinal data with a monotone missingness pattern, a robust
regression is fitted on the observed data at each visit point, incorporating
the observed historical information. Specifically, for the available
data at visit $s$ for $s=1,\cdots,t$, the imputation model parameter
estimate $\hat{\alpha}_{s-1}$ minimizes the loss function
\[
\sum_{i=1}^{n}(1-A_{i})R_{is}\rho(Y_{is}-H_{is-1}^{\T}\alpha_{s-1}).
\]
Here, $\rho(x)$ is the robust loss function. For example, the Huber
loss function is defined as $\rho(x)=0.5x^{2}\mathbb{I}(|x|\ensuremath{<l})+\left\{ l|x|-0.5l^{2}\right\} \mathbb{I}(|x|\ensuremath{\geq l})$,
where the constant $l>0$ controls the influence of the non-normal
data points and $\mathbf{\mathbb{I}}(\cdot)$ is an indicator function
\citep{huber1973robust}. When $l\rightarrow\infty$, the Huber-type
robust estimator is equivalent to the conventional LS estimator. We
also provide the definitions of the absolute loss and the $\varepsilon$-insensitive
loss in Section \ref{subsec:supp_cons} in the supplementary material. 

\begin{remark}[Tuning constant $l$ in the Huber loss function] 

The tuning constant $l$ in the Huber loss function mitigates the
impact of extreme values and heavy-tailed errors in the data. \citet{kelly1992robust}
argues the existence of trade-offs between the bias and variance in
the selection of the tuning constant. A small value of $l$ provides
more protections against non-normal values, yet suffers from the loss
of efficiency if the data is indeed normal. In practice, a common
recommendation of the tuning constant is $l=1.345\sigma$, where $\sigma$
is the standard deviation of the errors \citep{fox2002robust}. We
use the Huber loss function with this tuning parameter to get the
robust estimators throughout the simulation studies and real data
application. 

\end{remark}

While the estimator from the robust loss function provides protection
against extreme outliers in the response variables, it is not robust
against outliers in the covariates \citep{chang2018robust}, leading
to an imprecise estimation and a loss of efficiency. In longitudinal
data with the use of sequential regressions, the issue becomes more
profound, where the outcome is treated both as the response variable
in the current regression and as the covariate in the subsequent regression.
To deal with the outliers in the covariates, we utilize the idea in
\citet{carroll1993robustness} to down-weight the high leverage point
in the covariates via a robust Mahalanobis distance. Specifically,
for a $p$-dimensional covariate $X$, we calculate the robust Mahalanobis
distance as $d=(X-\mu)^{\T}V^{-1}(X-\mu)$, where $\mu$ is a robust
estimate of the center and $V$ is a robust estimate of the covariance
matrix. The trisquared redescending function is applied to form the
assigned weights as $w(u;\nu)=u\left\{ 1-(u/\nu)^{2}\right\} ^{3}\mathbb{I}(|u|\le\nu)$,
where $u=(d/\nu)^{1/2}$ and $\nu$ is a tuning parameter to control
the down-weight level. Therefore, in the sequential weighted robust
regression model, the robust estimate $\hat{\alpha}_{s-1}^{w}$ minimizes
the weighted loss
\begin{align}
\sum_{i=1}^{n}(1-A_{i})R_{is}w(H_{is-1};\nu_{s-1})\rho(Y_{is}-H_{is-1}^{\T}\alpha_{s-1}).\label{eq:seqrrw}
\end{align}

\begin{remark}[Tuning constants in the trisquared redescending function] 

When selecting the tuning parameter, \citet{carroll1993robustness}
used a fixed constant $\nu=8$ to illustrate a specific down-weight
behavior for cross-sectional studies. In longitudinal clinical trials
involving multiple weighted sequential regressions, the tuning parameter
$\nu_{s-1}$ can be selected via cross-validation at each visit point,
for $s=1,\cdots,t$. The main idea is to conduct a $K$-fold cross-validation
for the observed data at each visit point and determine the optimal
tuning parameter $\nu_{s-1}$ which minimizes the squared errors.
Specifically, we first partition the observed data at visit $s$ into
$K$ parts denoted as $P_{1},\cdots,P_{K}$. The part $P_{j}$ is
then left for the test, and the remaining $(K-1)$ folds are utilized
to learn the robust estimator $\hat{\alpha}_{s-1,-j}^{w}$, for $j=1,\cdots,K$.
The optimal $\nu_{s-1}$ minimizes the cross-validation sum of the
squared errors $\sum_{j=1}^{K}\sum_{i\in P_{j}}(Y_{is}-H_{is-1}^{\T}\hat{\alpha}_{s-1,-j}^{w})^{2}$. 

\end{remark}

After obtaining the robust estimates of the imputation model parameters,
we impute the missing components by their conditional outcome means
sequentially based on Assumptions \ref{assump:miss} and \ref{assump:j2r-mean}
and construct the imputed data $Y_{is}^{*}=R_{is}Y_{is}+(1-R_{is})H_{is-1}^{*\T}\hat{\alpha}_{s-1}^{w}$,
where $H_{is-1}^{*}=(X_{i}^{\T},Y_{i1}^{*},\cdots,Y_{is-1}^{*})^{\T}$.
Complete data analysis is then conducted on the imputed data, where
we again minimize the robust loss function to mitigate the impact
of outliers in the response variable. Note that since we assume that
there are no outliers in the baseline covariates, assigning the weights
to the loss function becomes unnecessary. Consider a general form
of the working model in the analysis step as
\begin{equation}
\mu(A,X\mid\gamma)=Ag(X;\gamma^{(0)})+h(X;\gamma^{(1)}),\label{eq:model_form}
\end{equation}
where $g(X;\gamma^{(0)})$ and $h(X;\gamma^{(1)})$ are integrable
functions bounded on compact sets, and $\gamma=(\gamma^{(0)\T},\gamma^{(1)\T})^{\T}$.
The robust estimator $\hat{\gamma}=(\hat{\gamma}^{(0)\T},\hat{\gamma}^{(1)\T})^{\T}$
can be found by minimizing the loss
\begin{equation}
\sum_{i=1}^{n}\rho\left\{ Y_{it}^{*}-\mu(A_{i},X_{i}\mid\gamma)\right\} ,\label{eq:analysis-1}
\end{equation}
and the resulting ATE estimator $\hat{\tau}$ is estimated by the
mean differences between the two groups as $\hat{\tau}=n^{-1}\sum_{i=1}^{n}{\color{black}g(X_{i};\hat{\gamma}^{(0)})}$.

The modeling form \eqref{eq:model_form} is commonly satisfied in
randomized trials when constructing the working model for analysis.
For example, the standard analysis model without the interaction term
between the treatment and the baseline covariates as $\mu(A,X\mid\gamma)=\gamma^{(0)}A+\gamma^{(1)\T}X$
gratifies this form when $g(X;\gamma^{(0)})=\gamma^{(0)}$ and $h(X;\gamma^{(1)})=\gamma^{(1)\T}X$;
a similar logic applies to the interaction model $\mu(A,X\mid\gamma)=\gamma^{(0)\T}AX+\gamma^{(1)\T}X$.
As we will elaborate in the next section, the ATE estimator $\hat{\tau}$
is analysis model-robust, in the sense that its asymptotic results
stay intact regardless of the specification of the analysis model.
The implementation of the proposed mean imputation-based robust method
is as follows.

\begin{enumerate}
\setlength{\itemindent}{1.5em}
\item[\textbf{Step 1}.] For the observed data in the control group,
fit the sequential weighted robust regression at each visit point
and get the sequential model parameter estimates $\hat{\alpha}_{s-1}^{w}$
by minimizing the weighted loss \eqref{eq:seqrrw} for $s=1,\cdots,t$.

\item[\textbf{Step 2}.] Impute missing data sequentially by the conditional
outcome mean according to Assumptions \ref{assump:miss} and \ref{assump:j2r-mean}
and obtain the imputed data $Y_{is}^{*}=R_{is}Y_{is}+(1-R_{is})H_{is-1}^{*\T}\hat{\alpha}_{s-1}^{w}$,
where $H_{is-1}^{*}=(X_{i}^{\T},Y_{i1}^{*},\cdots,Y_{is-1}^{*})^{\T}$
for $s=1,\cdots,t$.

\item[\textbf{Step 3}.] Set up an appropriate working model $\mu(A,X\mid\gamma)$
in the form \eqref{eq:model_form}, perform the complete data analysis
and get the ATE estimator $\hat{\tau}$ by minimizing the loss function
\eqref{eq:analysis-1}.

\end{enumerate}

The good theoretical properties of the ATE estimator along with a
linearization-based variance estimator are provided in the next section.

\section{Theoretical properties and analysis model robustness \label{sec:theo}}

We present the asymptotic theory of the ATE estimator in terms of
consistency and asymptotic normality along with a variance estimator
based on three robust loss functions as the Huber loss, the absolute
loss, and the $\varepsilon$-insensitive loss. To illustrate the theorems
in a straightforward way, we introduce additional notations. Denote
$\varphi(H_{is},\alpha_{s-1})=(1-A_{i})R_{is}w(H_{is-1};\nu_{s-1})\psi(Y_{is}-H_{is-1}^{\T}\alpha_{s-1})H_{is-1}$
as the function derived from minimizing the weighted loss function
\eqref{eq:seqrrw} in the imputation model, where $\psi(x)=\partial\rho(x)/\partial x$
is the derivative of the robust loss function, and $\alpha_{s-1,0}$
as the true parameter such that $\E\left\{ \varphi(H_{is},\alpha_{s-1})\mid H_{is-1}\right\} =0$.
Let $\hat{\alpha}^{w}=(\hat{\alpha}_{0}^{w\T},\cdots,\hat{\alpha}_{t-1}^{w\T})^{\T}$
be the combination of the model estimators from $t$ sequential regression
models in the imputation, and $\mathbb{\alpha}_{0}=(\alpha_{0,0}^{\T},\cdots,\alpha_{t-1,0}^{\T})^{\T}$
be the corresponding true model parameters. In terms of the components
in the analysis model, denote $\varphi_{a}(Z_{i},\gamma)=\psi\left\{ Y_{it}^{*}-\mu(A_{i},X_{i}\mid\gamma)\right\} \partial\mu(A_{i},X_{i}\mid\gamma)/\partial\gamma^{\T}$,
where $Z_{i}^{*}=(A_{i},X_{i}^{\T},Y_{it}^{*})^{\T}$ represents the
imputed data in the model, $\gamma_{0}$ is the true parameter such
that $\E\left\{ \varphi_{a}(Z_{i},\gamma)\right\} =0$, and $\tau_{0}$
is the true ATE such that $\tau_{0}=\E(Y_{it}\mid A=1)-\E(Y_{it}\mid A=0)$.
Suppose $\gamma^{(0)}$ is a $d_{0}$-dimensional vector, and $\gamma^{(1)}$
is a $d_{1}$-dimensional vector.

\begin{theorem}\label{thm:cons}

Under the regularity conditions listed in Section \ref{subsec:supp_cons}
in the supplementary material, the ATE estimator $\hat{\tau}\xrightarrow{\pr}\tau_{0}$
as the sample size $n\rightarrow\infty$, for $s=1,\cdots,t$.

\end{theorem} 

\begin{theorem}\label{thm:norm}

Under the regularity conditions listed in Section \ref{subsec:supp_norm}
in the supplementary material, as the sample size $n\rightarrow\infty$,
\[
\sqrt{n}(\hat{\tau}-\tau_{0})\xrightarrow{d}\mathcal{N}\Big(0,\mathbb{V}\left\{ V_{\tau,i}(\alpha_{0},\gamma_{0})\right\} \Big),
\]
where $V_{\tau,i}(\alpha_{0},\gamma_{0})=\left\{ \partial g(X_{i};\gamma_{0}^{(0)})/\partial\gamma^{\T}\right\} c^{\T}V_{\gamma,i}(\alpha_{0},\gamma_{0})$,
\begin{align*}
V_{\gamma,i}(\alpha_{0},\gamma_{0}) & =D_{\varphi}^{-1}\bigg[\varphi_{a}\left\{ Z_{i}^{*}(\beta_{t}),\gamma_{0}\right\} +\sum_{s=1}^{t}\E\left\{ R_{s-1}(1-R_{s})\frac{\partial\mu(A,X\mid\gamma_{0})}{\partial\gamma^{\T}}\frac{\partial\psi(e)}{\partial e}H_{s-1}^{\T}\right\} U_{t,s-1,i}(\alpha_{0})\bigg],
\end{align*}
$U_{t,s-1,i}(\alpha_{0})=\left(\I_{p+s-2},\alpha_{s-1,0}\right)U_{t,s,i}(\alpha_{0})+\left(\mathbf{0}_{p+s-2}^{\T},1\right)\beta_{t,s}q(H_{is},\alpha_{s-1,0})$
for $s<t$, $U_{t,t-1,i}(\alpha_{0})=q(H_{it},\alpha_{t-1,0})$, and
$q(H_{is},\alpha_{s-1,0})=\left[-\partial\E\left\{ \varphi(H_{is},\alpha_{s-1,0})H_{is-1}^{\T}\mid H_{is-1}\right\} /\partial\alpha_{s-1}^{\T}\right]^{-1}\varphi(H_{is},\alpha_{s-1,0}).$
Here, $c^{\T}=(\mathbf{I}_{d_{0}},\mathbf{0}_{d_{0}\times d_{1}})$
is a matrix where $\I_{d_{0}}$ is a $(d_{0}\times d_{0})$-dimensional
identity matrix and $\mathbf{0}_{d_{0}\times d_{1}}$ is a $(d_{0}\times d_{1})$-dimensional
zero matrix, $\mathbf{0}_{p+s-2}$ is a $(p+s-2)$-dimensional zero
vector, $D_{\varphi}=\partial\E\left[\varphi_{a}\left\{ Z_{i}^{*}(\beta_{t}),\gamma_{0}\right\} \right]/\partial\gamma^{\T}$
where $Z_{i}^{*}(\beta_{t})=\left(A_{i},X_{i}^{\T},Y_{it}^{*}(\beta_{t})\right){}^{\T}$
and $Y_{it}^{*}(\beta_{t})$ refers to the imputed value $Y_{it}^{*}$
based on the true imputation parameters $\beta_{t}=(\beta_{t,0}^{\T},\cdots\beta_{t,t-1}^{\T})^{\T}$
which satisfy
\[\begin{cases}
\beta_{t,t-1}=\alpha_{t-1,0} & \text{if \ensuremath{s=t}},\\
\beta_{t,s-1}=(\I_{p+s-2},\alpha_{s-1,0})(\I_{p+s-1},\alpha_{s,0})\cdots(\I_{p+t-3},\alpha_{t-2,0})\alpha_{t-1,0} & \text{if \ensuremath{s<t}},
\end{cases}\]
and $e_{i}=Y_{it}^{*}(\beta_{t})-\mu(A_{i},X_{i}\mid\gamma_{0})$.

\end{theorem} 

The asymptotic variance in Theorem \ref{thm:norm} motivates us to
obtain a linearization-based variance estimator by plugging in the
estimated values as
\[
\hat{\V}(\hat{\tau})=\frac{1}{n^{2}}\sum_{i=1}^{n}\left\{ V_{\tau,i}(\hat{\alpha}^{w},\hat{\gamma})-\bar{V}_{\tau}(\hat{\alpha}^{w},\hat{\gamma})\right\} ^{2},
\]
where $\bar{V}_{\tau}(\hat{\alpha}^{w},\hat{\gamma})=n^{-1}\sum_{i=1}^{n}V_{\tau,i}(\hat{\alpha}^{w},\hat{\gamma})$,
$V_{\tau,i}(\hat{\alpha}^{w},\hat{\gamma})=\left\{ \partial g(X_{i};\hat{\gamma}^{(0)})/\partial\gamma^{\T}\right\} c^{\T}V_{\gamma,i}(\hat{\alpha}^{w},\hat{\gamma})$,
\begin{align*}
V_{\gamma,i}(\hat{\alpha}^{w},\hat{\gamma}) & =D_{\varphi}^{-1}\bigg[\varphi_{a}\left(Z_{i}^{*},\hat{\gamma}\right)+\sum_{s=1}^{t}\left\{ \frac{1}{n}\sum_{i=1}^{n}R_{is-1}(1-R_{is})\frac{\partial\mu(A_{i},X_{i}\mid\hat{\gamma})}{\partial\gamma^{\T}}\frac{\partial\psi(\hat{e}_{i})}{\partial e_{i}}H_{is-1}^{\T}\right\} U_{t,s-1,i}(\hat{\alpha}^{w})\bigg],
\end{align*}
$U_{t,s-1,i}(\hat{\alpha}^{w})=\left(\I_{p+s-2},\hat{\alpha}_{s-1}^{w}\right)U_{t,s,i}(\hat{\alpha}^{w})+\left(\mathbf{0}_{p+s-2}^{\T},1\right)\hat{\beta}_{t,s}\hat{q}(H_{is},\hat{\alpha}_{s-1}^{w})$
for $s<t$, $U_{t,t-1,i}(\hat{\alpha}^{w})=\hat{q}(H_{it},\hat{\alpha}_{t-1}^{w})$,
and $\hat{q}(H_{is},\hat{\alpha}_{s-1}^{w})=\left[-n^{-1}\sum_{i=1}^{n}\left\{ \partial\varphi(H_{is},\hat{\alpha}_{s-1}^{w})/\partial\alpha_{s-1}^{\T}\right\} H_{is-1}^{\T}\right]^{-1}\varphi(H_{is},\hat{\alpha}_{s-1}^{w})$.
Also, $\hat{e}_{i}=Y_{it}^{*}-\mu(A_{i},X_{i}\mid\hat{\gamma})$,
$\hat{D}_{\varphi}=n^{-1}\sum_{i=1}^{n}\partial\varphi_{a}\left(Z_{i}^{*},\hat{\gamma}\right)/\partial\gamma^{\T}$,
and 
\[
\begin{cases}
\hat{\beta}_{t,t-1}=\hat{\alpha}_{t-1}^{w} & \text{if \ensuremath{s=t}},\\
\hat{\beta}_{t,s-1}=(\I_{p+s-2},\hat{\alpha}_{s-1}^{w})(\I_{p+s-1},\hat{\alpha}_{s}^{w})\cdots(\I_{p+t-3},\hat{\alpha}_{t-2}^{w})\hat{\alpha}_{t-1}^{w} & \text{if \ensuremath{s<t}},
\end{cases}
\]
for $s=1,\cdots,t$. Since the ATE estimator is asymptotically linear,
we can also use bootstrap to obtain a replication-based variance estimator.

We consider a specific working model as the interaction model for
analysis and present the asymptotic theories of the ATE estimator
in Section \ref{subsec:supp_interaction} in the supplementary material.
The interaction model is one of the most common models in the clinical
trials suggested in \citet{international2019addendum}, which is also
used in the simulation studies and real data application in the paper.

Theorems \ref{thm:cons} and \ref{thm:norm} extend the test robustness
\citep{rosenblum2009using} to the analysis model robustness in two
aspects. First, the robustness expands its plausibility from the hypothesis
test to the ATE estimation. Second, the robust estimator obtained
from minimizing the loss function further broadens the types of the
model estimator used in the analysis model. The resulting ATE estimator
via the robust loss remains consistent and has the identical asymptotic
normality even when the analysis model is misspecified. 

\section{Simulations \label{sec:simu}}

We conduct simulation studies to validate the finite-sample performance
of the proposed robust method. Consider a longitudinal clinical trial
with two treatment groups and five visits. Set the sample size for
each group as 500 and generate the data separately for each treatment.
The baseline covariates $X\in\mathbb{R}^{2}$ are a combination of
a continuous variable generated from the standard normal distribution
and a binary variable generated from a Bernoulli distribution with
the success probability of $0.3$. The longitudinal outcomes are generated
in a sequential manner, regressing on the historical information separately
for each group based on some specific distributions. The group-specific
data generating parameters are given in Section \ref{subsec:supp_simuset}
in the supplementary material. 

The missingness mechanism is set to be MAR with a monotone missingness
pattern. For the visit point $s$, if $R_{is'-1}=0$, then $R_{is'}=0$
for $s'=s,\cdots,t$; otherwise, let $R_{is}\mid\left(H_{is-1},A_{i}=a\right)\sim\text{Bernoulli}\left\{ \pi_{s}(a,H_{is-1})\right\} $.
We model the observed probability $\pi_{s}(a,H_{is-1})$ at visit
$s>1$ as a function of the observed information as $\text{logit}\left\{ \pi_{s}(a,H_{is-1})\right\} =\phi_{1a}+\phi_{2a}Y_{is-1}$,
where $\phi_{1a}$ and $\phi_{2a}$ are the tuning parameters for the observed
probabilities. The parameters are tuned to achieve the observed probability
around $0.8$ in each group.

We select the Huber loss function to obtain robust estimators for
its prevalence. Table \ref{table:sim}(a) summarizes the three methods
we aim to compare in the simulation studies. We apply distinct estimation
approaches for each method in the imputation and analysis models, along
with different imputation methods, where MI stands for the conventional
method used in longitudinal clinical trials and Robust stands for
our proposed method. LSE can be viewed as a transition from the conventional
MI method to the proposed robust method. In terms of the variance
estimation, Rubin's and bootstrap methods are compared for the MI
estimator while the linearization-based and bootstrap variance estimates
are compared for the mean imputation estimators. 

The simulation results are based on 10,000 Monte Carlo (MC) simulations
under $H_{0}:\tau=0$ and 1000 MC simulations under one specific alternative
hypothesis $H_{1}:\tau=\tau_{0}$, with the number of bootstrap replicates
$B=100$ and the imputation size $M=10$ for MI. The tuning parameter
of Huber robust regression is $l=1.345\sigma$, and the tuning parameters
of the sequential weighted robust models are $\nu_{s-1}=10$ for $s=1,\cdots,5$.
The imputation size $M$ and the tuning parameters $\nu_{s-1}$ do
not have a strong impact on the inferences (results are not shown). We assess
the estimators using the point estimate (Point est), the MC variance
(True var), the variance estimate (Var est), the relative bias of
the variance estimate computed by $\Big[\mathbb{E}\big\{\mathbb{\hat{V}}(\hat{\tau})\big\}-\mathbb{V}(\hat{\tau})\Big]/\mathbb{V}(\hat{\tau})$,
the coverage rate of $95\%$ confidence interval (CI), the type-1
error under $H_{0}$, the power under $H_{1}$, and the root mean
squared error (RMSE). We choose the $95\%$ Wald-type CI estimated
by $\big(\hat{\tau}-1.96\mathbb{\hat{V}}^{1/2}(\hat{\tau}),\hat{\tau}+1.96\mathbb{\hat{V}}^{1/2}(\hat{\tau})\big)$.

\subsection{Data with extreme outliers}

We first focus on the settings when the outcomes are generated sequentially
from the normal distribution with or without severe outliers. To produce
the outliers in the longitudinal outcomes, we randomly select 10
individuals from the top 30 completers with the highest outcomes at
the last visit point per group and multiply the original values
by three for all post-baseline outcomes. We also consider adding
extreme values only to one specific group and present the results
in Section \ref{subsec:supp_simtab} in the supplementary material. 

Table \ref{table:sim}(b) and the first two rows of Figure \ref{fig:sim} illustrate the
simulation results of the original data and the data with extreme
outliers under the normal distribution. Without the presence of outliers,
all methods produce unbiased point estimates. The robust method is
slightly less efficient compared to MI and LSE, as it has a larger
MC variance and a smaller power. For MI, Rubin's variance estimate
is conservative and inefficient, causing the coverage rate to be
far away from the empirical value and the power to be smaller, which
matches the observations detected in previous literature regarding
J2R in longitudinal clinical trials (e.g., \citealp{liu2016analysis,di2022}).
However, using bootstrap can fix the overestimation issue and produce
a reasonable coverage rate and power. When outliers exist, only
the robust method produces an unbiased point estimate, a well-controlled
type-1 error under $H_{0}$, and a satisfying coverage rate under
$H_{1}$ with a smaller RMSE.

\subsection{Data from a heavy-tailed distribution}

To assess the performance of the estimator from our proposed robust
method in heavy-tailed distributions, we generate the longitudinal
outcomes sequentially from a t-distribution with the degrees of freedom
as 5 in time order. The detailed setup of the data-generating process
is also given in Section \ref{subsec:supp_simuset} in the supplementary material. 

Table \ref{table:sim}(c) and the last row of Figure \ref{fig:sim} show the simulation
results. All the methods result in unbiased point
estimates. The robust method produces the ATE estimator with the smallest
MC variance, indicating the superiority of Huber robust regression
under a heavy-tailed distribution. The linearization-based variance
estimates behave similarly to the bootstrap variance estimates for
the two mean imputation-based methods, with comparable coverage rates
and powers.

\begin{table}
\caption{Summary of the simulation methods and results.}\label{table:sim}
 \centering \subfloat[\large{Different estimation and imputation approach in the four methods used
for comparison.}\label{tab:sum_method-1}]{
\centering{}%
\begin{tabular}{cccc}
\hline 
Method & Imputation model & Imputation method & Analysis model\tabularnewline
\hline 
MI & LS & MI & LS\tabularnewline
LSE & Weighted Huber regression & Mean imputation & LS\tabularnewline
Robust & Weighted Huber regression & Mean imputation & Huber regression\tabularnewline
\hline 
\end{tabular}}
\vspace{3ex}
\subfloat[\large{Simulation results under the normal distributions without or with
extreme outliers. Here the true value $\tau=71.18\%$.}\label{table:extreme-1}]{
\centering{}\scalebox{1}{ \resizebox{\textwidth}{!}{%
\begin{tabular}{>{\raggedright}p{0.1\textwidth}>{\centering}p{0.1\textwidth}cccccccccccccc}
\hline 
 &  & Point est & True var & \multicolumn{2}{c}{Var est} &  & \multicolumn{2}{c}{Relative bias} &  & \multicolumn{2}{c}{Coverage rate} &  & \multicolumn{2}{c}{Power} & RMSE\tabularnewline
Case & \multicolumn{1}{c}{Method} & ($\times10^{-2}$) & ($\times10^{-2}$) & \multicolumn{2}{c}{($\times10^{-2}$)} &  & \multicolumn{2}{c}{($\%$)} &  & \multicolumn{2}{c}{($\%$)} &  & \multicolumn{2}{c}{($\%$)} & ($\times10^{-2}$)\tabularnewline
 &  &  &  & $\hat{V}_{1}$ & $\hat{V}_{\text{Boot}}$ &  & $\hat{V}_{1}$ & $\hat{V}_{\text{Boot}}$ &  & $\hat{V}_{1}$ & $\hat{V}_{\text{Boot}}$ &  & $\hat{V}_{1}$ & $\hat{V}_{\text{Boot}}$ & \tabularnewline
\hline 
\multirow{3}{0.1\textwidth}{No outliers} & \multicolumn{1}{c}{MI} & 70.89 & 3.02 & 5.35 & 3.22 &  & 77.00 & 6.39 &  & 98.90 & 95.00 &  & 92.80 & 98.00 & 17.38\tabularnewline
 & LSE & 70.93 & 3.03 & 3.25 & 3.15 &  & 7.14 & 4.04 &  & 95.40 & 94.80 &  & 97.70 & 97.80 & 17.40\tabularnewline
 & Robust & \multirow{1}{*}{70.09} & \multirow{1}{*}{3.26} & \multirow{1}{*}{3.41} & \multirow{1}{*}{3.38} &  & \multirow{1}{*}{4.75} & \multirow{1}{*}{3.73} &  & \multirow{1}{*}{95.00} & \multirow{1}{*}{94.20} &  & \multirow{1}{*}{96.70} & \multirow{1}{*}{96.70} & \multirow{1}{*}{18.07}\tabularnewline
\hline 
\multirow{3}{0.1\textwidth}{Outliers in both groups} & \multicolumn{1}{c}{MI} & 77.42 & 3.92 & 12.42 & 8.97 &  & 216.36 & 128.51 &  & 99.70 & 99.10 &  & 65.90 & 83.60 & 20.76\tabularnewline
 & LSE & 74.02 & 4.18 & 6.18 & 5.94 &  & 47.85 & 42.18 &  & 98.30 & 97.80 &  & 89.80 & 91.40 & 20.62\tabularnewline
 & Robust & \multirow{1}{*}{72.34} & \multirow{1}{*}{3.44} & \multirow{1}{*}{3.53} & \multirow{1}{*}{3.49} &  & \multirow{1}{*}{2.75} & \multirow{1}{*}{1.40} &  & \multirow{1}{*}{95.00} & \multirow{1}{*}{94.60} &  & \multirow{1}{*}{97.00} & \multirow{1}{*}{96.60} & \multirow{1}{*}{18.57}\tabularnewline
\hline 
\end{tabular}} }}
\vspace{3ex}
\subfloat[\large{Simulation results under the t-distribution. Here the true value $\tau=68.09\%$.}\label{table:mvt h1-1}]{\centering{} \scalebox{1}{ \resizebox{\textwidth}{!}{%
\begin{tabular}{>{\centering}p{0.1\textwidth}cccccccccccccc}
\toprule 
 & Point est & True var & \multicolumn{2}{c}{Var est} &  & \multicolumn{2}{c}{Relative bias} &  & \multicolumn{2}{c}{Coverage rate} &  & \multicolumn{2}{c}{Power} & RMSE\tabularnewline
Method & ($\times10^{-2}$) & ($\times10^{-2}$) & \multicolumn{2}{c}{($\times10^{-2}$)} &  & \multicolumn{2}{c}{($\%$)} &  & \multicolumn{2}{c}{($\%$)} &  & \multicolumn{2}{c}{($\%$)} & ($\times10^{-2}$)\tabularnewline
 &  &  & $\hat{V}_{1}$ & $\hat{V}_{\text{Boot}}$ &  & $\hat{V}_{1}$ & $\hat{V}_{\text{Boot}}$ &  & $\hat{V}_{1}$ & $\hat{V}_{\text{Boot}}$ &  & $\hat{V}_{1}$ & $\hat{V}_{\text{Boot}}$ & \tabularnewline
\midrule
\multicolumn{1}{c}{MI} & 70.42 & 3.00 & 5.35 & 3.16 &  & 78.07 & 5.27 &  & 99.30 & 95.40 &  & 90.90 & 97.60 & 17.47\tabularnewline
LSE & 70.54 & 2.90 & 3.21 & 3.06 &  & 10.49 & 5.43 &  & 96.30 & 95.30 &  & 98.00 & 98.10 & 17.20\tabularnewline
Robust & \multirow{1}{*}{69.81} & \multirow{1}{*}{2.72} & \multirow{1}{*}{2.90} & \multirow{1}{*}{2.81} &  & \multirow{1}{*}{6.38} & \multirow{1}{*}{3.11} &  & \multirow{1}{*}{94.90} & \multirow{1}{*}{94.80} &  & \multirow{1}{*}{98.30} & \multirow{1}{*}{98.40} & \multirow{1}{*}{16.58}\tabularnewline
\bottomrule
\end{tabular}} }}

\noindent\begin{minipage}[t]{1\columnwidth}%
\vspace{-1ex}

{\footnotesize{}{\raggedright }$\hat{V}_{1}${\footnotesize{} denotes
the variance estimate obtained by Rubin's rule in MI and linearization
in mean imputation-based methods; }$\hat{V}_{\text{Boot}}${\footnotesize{}
denotes the bootstrap variance estimates.}{\footnotesize\par}

{\footnotesize{}}}{\footnotesize\par}%
\end{minipage}
\end{table}

\begin{figure}
\caption{Plot for the simulation results under different distributions.\label{fig:sim}}

\centering{}\includegraphics[scale=0.5]{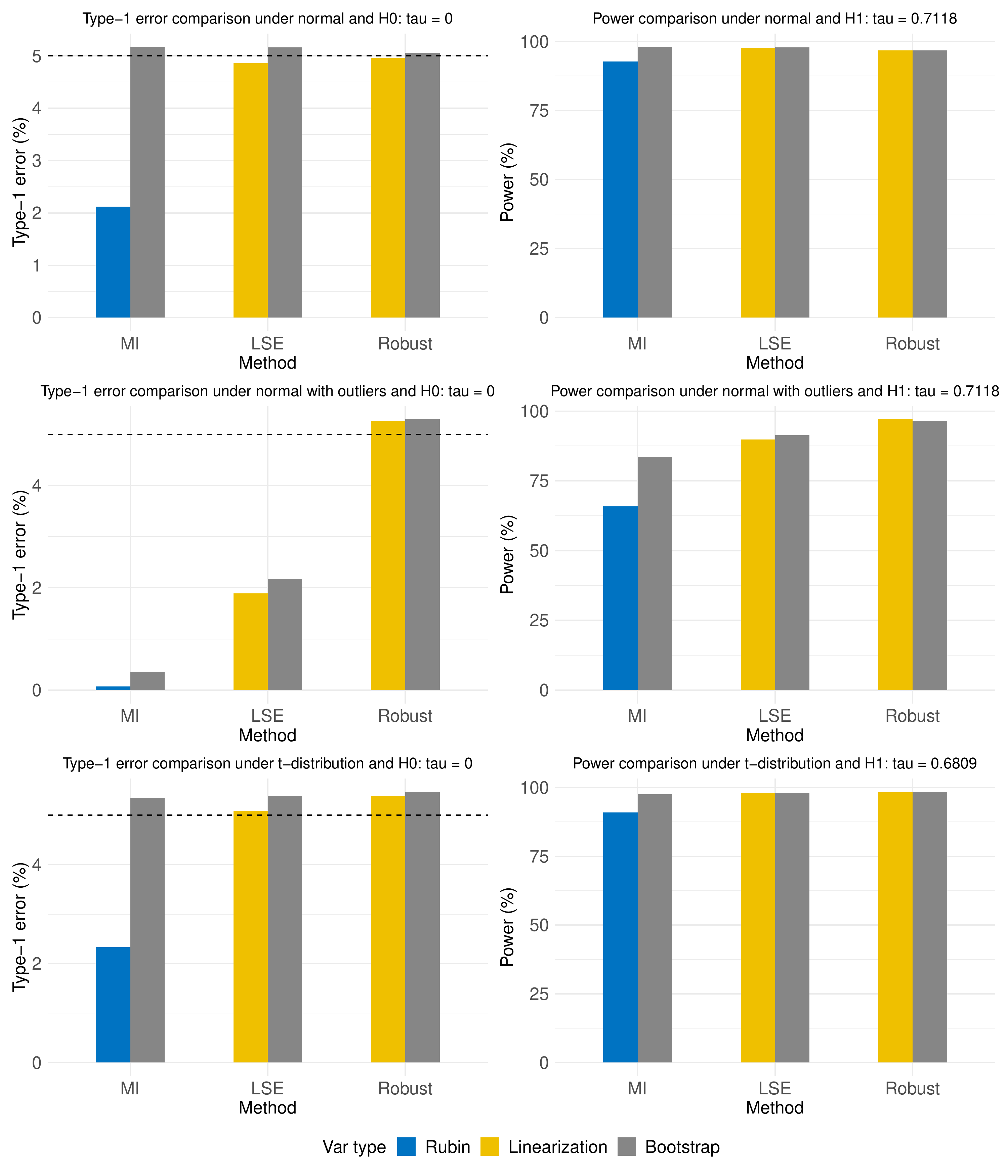}
\end{figure}

The overall simulation results indicate a recommendation of the proposed
robust approach with the linearization-based variance estimation to
obtain unbiased point estimates and save computation time. The advocated
method works well in terms of consistency, well-controlled type-1
errors, higher powers under $H_{1}$, and smaller RMSEs. Even under
the normality assumption, our proposed method has comparable performance
as the conventional MI method, with only a slight loss in the power.
When encountering a heavy-tailed distribution or extreme outliers,
the proposed method outperforms with more reasonable coverage rates
and higher powers. Similar interpretations apply to the simulation
results under $H_{0}$ given in Section \ref{subsec:supp_simtab}
in the supplementary material.

\section{Estimating effects of HIV-1 reverse transcriptase inhibitors \label{sec:app}}

We now apply our proposed robust method to the motivating example
introduced in Section \ref{sec:moti_exmp}. The primary goal is to
assess the ATE between the two arms at the study endpoint under the
J2R condition. The results of the normality test and the symmetry test proposed by
\citet{miao2006new} in Figure \ref{fig:diagnosis}
indicate that the data are symmetrically distributed without severe
outliers, yet suffer from a heavy tail that deviates from normality.
MI, mean imputation with LS estimators, and the proposed robust method
using the Huber loss function are compared with respect to the point
estimation, the variance estimation based on Rubin's variance estimator
or the linearization-based variance estimator, Wald-type $95\%$ CI
and CI length. For MI, the imputation size is $M=100$. The tuning
parameters for the weights in the robust method are selected via cross-validation,
with the procedure described in Section \ref{sec:supp_real} in the
supplementary material.

Table \ref{table:real} shows the analysis results of the group means
and the ATE under J2R. MI uses the sequential linear regressions estimated
by the LS estimators for the imputation model, resulting in different
point estimates compared to other mean imputation-based methods, where
the imputation model is obtained via robust regressions. Using LS
or Huber loss in the analysis model also has a slight difference in
the estimation because of the heavy tail. While the conventional MI
method may contaminate the inference when the data deviates from the
normal distribution, the proposed robust method preserves an unbiased
estimate and a narrower CI, which coincides with the conclusions drawn
from the simulation studies. All the implemented methods show a statistically
significant treatment effect under J2R, uncovering the superiority
of triple therapies. 

\begin{table}[!htbp]
\centering \caption{Analysis of the repeated CD4 count data under J2R.}
\label{table:real} \scalebox{1}{
\centering{}%
\begin{tabular}{>{\centering}p{0.12\textwidth}ccccc}
\toprule 
Variable & \multicolumn{1}{c}{Method} & Point est & \multicolumn{1}{c}{$95\%$ CI} &  & \multicolumn{1}{c}{CI length}\tabularnewline
\midrule 
\multirow{3}{0.12\textwidth}{\centering{Mean arm 1}} & \multicolumn{1}{c}{MI} & -0.68 & (-0.84, -0.53) &  & 0.31\tabularnewline
 & LSE & -0.54 & (-0.67, -0.42) &  & 0.25\tabularnewline
 & Robust & -0.53 & (-0.64, -0.41) &  & 0.23\tabularnewline
\midrule 
\multirow{3}{0.12\textwidth}{\centering{Mean arm 2}} & MI & -0.39 & (-0.55, -0.24) &  & 0.31\tabularnewline
 & \multicolumn{1}{c}{LSE} & -0.23 & (-0.35, -0.11) &  & 0.24\tabularnewline
 & Robust & -0.27 & (-0.38, -0.16) &  & 0.22\tabularnewline
\midrule
\multirow{3}{0.12\textwidth}{\centering{Difference}} & MI & 0.29 & (0.07, 0.51) &  & 0.44\tabularnewline
 & LSE & 0.31 & (0.20, 0.41) &  & 0.21\tabularnewline
 & Robust & 0.26 & (0.16, 0.35) &  & 0.19\tabularnewline
\bottomrule
\end{tabular}} 
\end{table}

\section{Conclusion \label{sec:conclusion}}

The non-normality issue frequently occurs in longitudinal clinical
trials due to extreme outliers or heavy-tailed errors. With growing
attention to evaluating the treatment effect with an MNAR missingness mechanism,
we establish a robust method with the weighted robust regression and
mean imputation under J2R for the longitudinal data, without the reliance
on parametric models. The weighted robust regression provides double-layer
protection against non-normality in both the covariates and the response
variable, therefore ensuring a valid imputation model estimator. Mean
imputation and the subsequent robust analysis model further guarantee
a valid ATE estimator with good theoretical properties. The proposed
method also enjoys the analysis model robustness property, in the
sense that the consistency and asymptotic normality of the ATE estimator
are satisfied even when the analysis model is incorrectly specified.

The symmetry error distribution, which is an essential assumption
in the robust regression using the robust loss, must be satisfied
in order to obtain a grounded inference for the ATE. It may not always
be the case in practice. When encountering skewed distributions with
asymmetric noises, biases and imprecisions may be detected in our
proposed robust method. \citet{takeuchi2002robust} provide a novel
robust regression method motivated by data mining to handle asymmetric
tails and obtain reasonable mean-type estimators. The extension of
the proposed robust method may be plausible by replacing the robust
regression with their proposed regression model.

While we focus solely on a monotone missingness pattern throughout
the development of the robust method, intermittent missingness is
also ubiquitous in longitudinal clinical trials. To handle intermittent
missing data with a non-ignorable missingness mechanism when the outcomes
are non-normal, \citet{elashoff2012robust} suggest incorporating
the Huber loss function in the pseudo-log-likelihood expression to
obtain robust inferences. It is possible to extend our proposed robust
method using their idea. We leave it as a future working direction.

\section*{{Acknowledgements}}
Yang is partially supported by the NSF grant DMS 1811245, NIA grant 1R01AG066883, and NIEHS grant 1R01ES031651. 

\section*{Supplementary material}

The supplementary material contains the proofs and more technical details.

\bibliographystyle{Chicago}
\bibliography{DR_PSACE_ref}

\newpage{} 
\begin{center}
\textbf{\Large{}Supplementary material for "Robust analyses for longitudinal clinical trials with missing
and non-normal continuous outcomes" by Liu et al.}{\Large{} }{\Large\par}
\par\end{center}

\pagenumbering{arabic} 
\renewcommand*{\thepage}{S\arabic{page}}

\setcounter{lemma}{0} 
\global\long\def\thelemma{\textup{S}\arabic{lemma}}%
\setcounter{equation}{0} 
\global\long\def\theequation{S\arabic{equation}}%
\setcounter{section}{0} 
\global\long\def\thesection{S\arabic{section}}%
\global\long\def\thesubsection{S\arabic{section}.\arabic{subsection}}%
\setcounter{table}{0} 
\global\long\def\thetable{\textup{S}\arabic{table}}%
\setcounter{figure}{0} 
\global\long\def\thefigure{\textup{S}\arabic{figure}}%
\setcounter{thm}{0}
\global\long\def\thethm{\textup{S}\arabic{thm}}%
\setcounter{corollary}{0}
\global\long\def\thecorollary{\textup{S}\arabic{corollary}}%

The supplementary material contains technical details, additional
simulations, and real-data application results. Section \ref{sec:supp_infer}
gives the regularity conditions and the proof of the model-robust
ATE estimator obtained from the proposed robust method in terms of
consistency and asymptotic normality, and provides an example of the working
model for illustration and extensions to other robust loss functions.
Section \ref{sec:supp_seqreg} provides the sequential regression procedure.
Section \ref{sec:supp_sim} shows additional simulation results when
the data is incorporated from outliers or different data-generating
distributions. Section \ref{sec:supp_real} adds additional notes
on the real data.

\section{Asymptotic results for the ATE estimator \label{sec:supp_infer}}

In this section, we present the asymptotic properties of the ATE estimator
$\hat{\tau}$ obtained from the proposed robust method in terms of
consistency and asymptotic normality. To begin with, we explore the
asymptotic properties of $\hat{\alpha}_{s-1}^{w}$ based on the observed
data at visit $s$ in the control group that minimizes the weighted
loss function (1) in the main text. Since the Huber loss is strongly
convex, minimizing the loss function is equivalent to find the root
of the first derivative
\begin{align*}
\sum_{i=1}^{n}(1-A_{i})R_{is}w(H_{is-1};\nu_{s-1})\psi(Y_{is}-H_{is-1}^{\T}\alpha_{s-1})H_{is-1}=0.
\end{align*}
We give the consistency result for the robust estimator $\hat{\alpha}_{s-1}^{w}$
in the following lemma.

\begin{lemma}\label{lemma:seq_cons}

Assume the following regularity conditions:

\begin{enumerate}

\item[C1.] There exists a unique $\alpha_{s-1,0}$ lying in the interior
of the Euclidean parameter space $\Theta$, such that the distribution
of the observed regression errors $(Y_{s}-H_{s-1}^{\text{\T}}\alpha_{s-1,0})$
is symmetric around 0.

\item[C2.] $\E\left\{ \psi(Y_{s}-H_{s-1}^{\T}\alpha_{s-1})\mid H_{s-1}\right\} $
is dominated by an integrable function $g(H_{s-1})$ for all $H_{s-1}\subset\mathbb{R}^{p+s-1}$
and $\alpha_{s-1}$ with respect to the conditional distribution function
$f(Y_{s}\mid H_{s-1},\alpha_{s-1})$.

\end{enumerate} Then, the estimator $\hat{\alpha}_{s-1}^{w}\xrightarrow{\pr}\alpha_{s-1,0}$
as the sample size $n\rightarrow\infty$, for $s=1,\cdots,t$.

\end{lemma}

\begin{proof} Note that by the definition of Huber function, $\psi(Y_{s}-H_{s-1}^{\T}\alpha_{s-1})$
is a continuous function for $\alpha_{s-1}$ and a measurable function
for $H_{s}$. By the regularity condition C2, it satisfies the conditions
for Theorem 2 in \citet{jennrich1969asymptotic}. Thus 
\[
\psi(Y_{s}-H_{s-1}^{\T}\alpha_{s-1})\xrightarrow{a.s.}\E\big\{\psi(Y_{s}-H_{s-1}^{\T}\alpha_{s-1})\mid H_{k-1}\big\}\text{ uniformly for }\forall\alpha_{s-1}\in\Theta.
\]

In the weighted sequential robust regression, at $s$th visit point,
the true value $\beta_{s-1,0}$ is the unique solution such that $\E\{(1-A)R_{s}w(H_{s-1};q_{s-1})\psi(Y_{s}-H_{s-1}^{\T}\alpha_{s-1})H_{s-1}\mid H_{s-1}\}=0$
since it is a randomized trial and by Assumption 1,
\begin{align*}
 & \E\big\{(1-A)R_{s}w(H_{s-1};\nu_{s-1})\psi(Y_{s}-H_{s-1}^{\T}\alpha_{s-1})H_{s-1}\mid H_{s-1}\big\}\\
= & \E(1-A)\E(R_{s}\mid H_{s-1})\E\big\{ w(H_{s-1};\nu_{s-1})\psi(Y_{s}-H_{s-1}^{\T}\alpha_{s-1})H_{s-1}\mid H_{s-1}\big\}\\
= & \E(1-A)\E(R_{s}\mid H_{s-1})\E\big\{\psi(Y_{s}-H_{s-1}^{\T}\alpha_{s-1})\mid H_{s-1}\big\} w(H_{s-1};\nu_{s-1})H_{s-1}.
\end{align*}
By the regularity condition C1, $(Y_{s}-H_{s-1}^{\T}\alpha_{s-1,0})$
is symmetric around 0 indicates that $\E\big\{\psi(Y_{s}-H_{s-1}^{\T}\alpha_{s-1})\mid H_{s-1}\big\}=0$.
Based on the regularity conditions C1 and C2, we apply Theorem 7.1
in \citet{boos2013essential} and get $\hat{\alpha}_{s-1}^{w}\xrightarrow{\pr}\alpha_{s-1,0}$,
for $s=1,\cdots,t$. \end{proof}

After obtaining the imputation model estimate $\hat{\alpha}_{s-1}^{w}$
for each visit point, we conduct sequential mean imputation to the
missing components and get $Y_{s}^{*}=R_{s}Y_{s}+(1-R_{s})H_{s-1}^{*\T}\hat{\alpha}_{s-1}^{w}$,
where $H_{s-1}^{*}=(X^{\T},Y_{1}^{*},\cdots,Y_{s-1}^{*})^{\T}$ for
$s=1,\cdots,t$. Denote the true imputation model parameter needed
for imputing the outcome at visit $t$ when the individual drops out
at visit $s$ as $\beta_{t,s-1}$, such that $\E(Y_{t}\mid H_{s-1},R_{s-1}=1,R_{s}=0,A=a)=H_{s-1}^{\T}\beta_{t,s-1}$
for $t\geq s$. The following lemma characterizes the relationship
between $\beta_{t,s-1}$ and the sequential imputation model parameters
$\alpha_{s-1},\cdots,\alpha_{t-1}$ .

\begin{lemma}\label{lemma:beta-alpha}

Under the regularity conditions C1 and C2, the parameter $\beta_{t,s-1}$
in formula \eqref{eq:impute_value} relates to the sequential imputation
model parameters $\alpha_{s-1,0},\cdots,\alpha_{t-1,0}$ in the following
way:

\begin{equation}
\begin{cases}
\beta_{t,t-1}=\alpha_{t-1,0} & \text{if \ensuremath{s=t}},\\
\beta_{t,s-1}=(\I_{p+s-2},\alpha_{s-1,0})(\I_{p+s-1},\alpha_{s,0})\cdots(\I_{p+t-3},\alpha_{t-2,0})\alpha_{t-1,0} & \text{if \ensuremath{s<t}}.
\end{cases}\label{eq:beta-alpha}
\end{equation}

\end{lemma} 

\begin{proof} The regularity condition C1 implies that the distribution
of the errors $(Y_{s}-H_{s-1}^{\T}\alpha_{s-1,0})$ is symmetric,
we have $\E(Y_{s}\mid H_{s-1})=H_{s-1}^{\T}\alpha_{s-1,0}$ for $s=1,\cdots,t$.
If the individual in group $a$ drops out at visit $t$, the missing
value at visit $t$ is imputed by
\begin{align*}
H_{t-1}^{\T}\beta_{t,t-1} & =\E(Y_{t}\mid H_{t-1},R_{t-1}=1,R_{t}=0,A=a)\\
 & =\E(Y_{t}\mid H_{t-1},A=0;\alpha_{t-1,0})\text{ (By Assumption 2)}\\
 & =H_{t-1}^{\T}\alpha_{t-1,0}\text{ (By C1)}.
\end{align*}
if using the true imputation model parameter. 

We then prove formula \eqref{eq:beta-alpha} by induction. Suppose
the result holds for the individual who drops out at visit $s$, i.e.,
we impute the value at the last visit point by $\E\left(Y_{t}\mid H_{s-1},A=0\right)=H_{s-1}^{\T}\beta_{t,s-1}$.
Then for the one in group $a$ who drops out at visit $s-1$, the
missing outcome at visit $s$ is imputed by
\begin{align*}
H_{s-2}^{\T}\beta_{t,s-2} & =\E(Y_{t}\mid H_{s-2},R_{s-2}=1,R_{s-1}=0,A=a)\\
 & =\E(Y_{t}\mid H_{s-2},A=0)\text{ (By Assumption 2)}\\
 & =\E\left\{ \E\left(Y_{t}\mid H_{s-1},A=0\right)\mid H_{s-2},A=0\right\} \\
 & =\E\left(H_{s-1}^{\T}\beta_{t,s-1}\mid H_{s-2},A=0\right)\\
 & =\left(H_{s-2}^{\T},\E(Y_{s-1}\mid H_{s-2},A=0)\right)^{\T}\beta_{t,s-1}\\
 & =H_{s-2}^{\T}\left(\I_{p+s-2},\alpha_{s-2,0}\right)\beta_{t,s-1}.
\end{align*}
Then we have 
\begin{align*}
\beta_{t,s-2} & =\left(\I_{p+s-2},\alpha_{s-2,0}\right)\beta_{t,s-1}\\
 & =\left(\I_{p+s-2},\alpha_{s-2,0}\right)(\I_{p+s-2},\alpha_{s-1,0})(\I_{p+s-1},\alpha_{s,0})\cdots(\I_{p+t-3},\alpha_{t-2,0})\alpha_{t-1,0},
\end{align*}
which completes the proof.\end{proof}

Lemma \ref{lemma:beta-alpha} suggests an estimator of $\beta_{t,s-1}$
by plugging in the sequential imputation model parameter estimates
$\hat{\alpha}_{s-1}^{w},\cdots,\hat{\alpha}_{t-1}^{w}$ in formula
\eqref{eq:beta-alpha}. Set $R_{0}=1$, we can rewrite the imputed
value $Y_{t}^{*}$ at the last visit point based on the observed history,
the dropout pattern, and the estimated imputation model parameters
as
\begin{equation}
Y_{t}^{*}=R_{t}Y_{t}+\sum_{s=1}^{t}R_{s-1}(1-R_{s})H_{s-1}^{\T}\hat{\beta}_{t,s-1},\label{eq:impute_value}
\end{equation}
where $\hat{\beta}_{t,s-1}$ is the estimate of $\beta_{t,s-1}$.
We proceed to prove the consistency of the ATE estimator $\hat{\tau}$.

\subsection{Proof of Theorem 1 \label{subsec:supp_cons}}

To illustrate the dependence of the imputed value $Y_{t}^{*}$ with
the imputed parameter estimates $\hat{\beta}_{t}:=(\hat{\beta}_{t,0},\cdots,\hat{\beta}_{t,t-1})^{\T}$,
we rewrite $Y_{t}^{*}$ as $Y_{t}^{*}(\hat{\beta}_{t})$ and $Z^{*}$
as $Z^{*}(\hat{\beta}_{t})$. Denote $\beta_{t}:=(\beta_{t,0},\cdots,\beta_{t,t-1})^{\T}$
as the true value. 

We do not assume the correct model form for the analysis model, instead,
we give a wide range of models of the form (2) in the main text.
When a symmetric error distribution is imposed, we write the model
as
\begin{align}
Y_{t}^{*}(\hat{\beta}_{t}) & =\mu(A,X\mid\gamma_{0})+\varepsilon,\nonumber \\
 & =\left(A-\pi\right)g(X;\gamma_{0}^{(0)})+\tilde{h}(X)+\varepsilon\label{eq:work_model}
\end{align}
where $\varepsilon$ is the error term and is symmetric around $0$,
and $\tilde{h}(X)=\pi g(X;\gamma_{0}^{(0)})+h(X;\gamma_{0}^{(1)})$.
Under the symmetric error assumption, $\gamma_{0}^{(0)}$ satisfies
that $\E\left\{ g(X;\gamma^{(0)})\right\} =\E(Y_{t}\mid A=1)-\E(Y_{t}\mid A=0)=\tau_{0}$.

The robust estimator $\hat{\gamma}=(\hat{\gamma}^{(0)\T},\hat{\gamma}^{(1)\T})^{\T}$
is obtained from
\[
\left(\hat{\gamma}^{(0)},\hat{\gamma}^{(1)}\right)=\arg\min\frac{1}{n}\sum_{i=1}^{n}\rho\left\{ Y_{it}^{*}(\hat{\beta_{t}})-\mu(A,X\mid\gamma)\right\} .
\]
We now want to show that $\hat{\gamma}^{(0)}$ obtained from the robust
loss function satisfies that $\hat{\tau}=\sum_{i=1}^{n}g(X_{i};\allowbreak \hat{\gamma}^{(0)})\xrightarrow{\pr}\tau_{0}$.
Before restating Theorem 1 in the main text with technical details, we
first give the definitions of the two robust loss functions as the
absolute loss and the $\varepsilon$-insensitive loss. The absolute
loss function is defined as $\rho_{a}(x)=|x|$. The $\varepsilon$-insensitive
loss is defined as $\mathcal{L}_{\varepsilon}(x)=\max\left\{ |x|-\varepsilon,0\right\} $,
where the constant $\varepsilon>0$ provides a tolerance margin where
no penalties are given \citep{smola2004tutorial}.

\begin{thm}

Under the regularity conditions C1 and C2, and assume the following
regularity conditions holds for $a=0,1$:

\begin{enumerate}

\item[C3.] Given the baseline covariates, the error term is conditionally
independent with the treatment variable, i.e., $\varepsilon\indep A\mid X$.

\item[C4.] For any $\zeta$, the term $K_{1}:=\tilde{h}(X)+\varepsilon_{i}-\tilde{h}(X;\zeta)$
has an expectation and a finite second moment, i.e., $\E|K_{1}|<\infty$
and $\E(K_{1}^{2})<\infty$, where $\tilde{h}(X)=\pi g(X;\gamma_{0}^{(0)})+h(X)$,
$\tilde{h}(X;\zeta)$ is a parametric model of $\tilde{h}(X)$, and
$\gamma_{0}^{(0)}$ is the unique solution such that $\E\left\{ g(X;\gamma^{(0)})\right\} =\tau_{0}$.

\item[C5.] For any $\gamma^{(0)}$, the term $K_{2}:=\left(A-\pi\right)\left\{ g(X;\gamma^{(0)})-g(X;\gamma_{0}^{(0)})\right\} $
has an expectation and a finite second moment, i.e., $\E|K_{2}|<\infty$
and $\E(K_{2}^{2})<\infty$.

\item[C6.] For any $f_{j}(X)$, $\pr\left\{ X:g(X;\gamma^{(0)})-g(X;\gamma_{0}^{(0)})\neq0\right\} >0$,
for $\forall\gamma^{(0)}\neq0$.

\item[C7.] The error term $\varepsilon\mid X=x$ has a non-zero density
function.

\item[C8.] The function $G_{2}(\zeta):=\E\left[\rho(K_{1})\right]$
has a unique global minimizer $\zeta^{*}$.

\end{enumerate} Then, the ATE estimator $\hat{\tau}\xrightarrow{\pr}\tau_{0}$
as the sample size $n\rightarrow\infty$, for $s=1,\cdots,t$.

\end{thm}

\begin{proof} We begin the proof by rewriting the working model (2).
Note that
\begin{align}
\mu(A,X\mid\gamma) & =Ag(X;\gamma^{(0)})+h(X;\gamma^{(1)})\nonumber \\
 & =\left(A-\pi\right)g(X;\gamma^{(0)})+\pi g(X;\gamma^{(0)})+h(X;\gamma^{(1)})\nonumber \\
 & =\left(A-\pi\right)g(X;\gamma^{(0)})+\tilde{h}(X;\zeta),\label{eq:contrast}
\end{align}
where $\tilde{h}(X;\zeta)=\pi g(X;\gamma^{(0)})+h(X;\gamma^{(1)})$
and $\zeta$ combines the parameters $\gamma^{(0)}$ and $\gamma^{(1)}$.
We are interested in estimating $g(X;\gamma^{(0)})$, as it is the
only part that connects with the ATE estimation. We want to prove
that $\hat{\gamma}^{(0)}$ obtained from minimizing the Huber loss
function satisfies that $\hat{\gamma}^{(0)}\xrightarrow{\pr}\gamma_{0}^{(0)}$,
regardless of the model specification. 

By Lemmas \ref{lemma:seq_cons} and \ref{lemma:beta-alpha}, we have
$\hat{\beta}_{t}\xrightarrow{\pr}\beta_{t}$. Follow the similar proof
in Lemma \ref{lemma:seq_cons}, by continuous mapping theorem, we
have $Y_{t}^{*}(\hat{\beta}_{t})=Y_{t}^{*}(\beta_{t})+o_{\pr}(1)=\left(A-\pi\right)g(X;\gamma_{0}^{(0)})+\tilde{h}(X)+\varepsilon+o_{\pr}(1)$.
Therefore, minimizing the loss function $n^{-1}\sum_{i=1}^{n}\rho\big\{ Y_{i}(\hat{\beta}_{t})-\mu(A_{i},X_{i}\mid\gamma)\big\}$
is asymptotically equivalent to minimizing $n^{-1}\sum_{i=1}^{n}\rho\big\{ Y_{i}(\beta_{t})-\mu(A_{i},X_{i}\mid\gamma)\big\}$.

We then follow the proof in \citet{xiao2019robust} to verify the
consistency of the estimator under $H_{0}$ using the Huber loss function,
the absolute loss function, or the $\varepsilon$-insensitive loss.

\paragraph{(i) For the Huber loss,}

denote $L_{n}(\gamma^{(0)},\zeta)=n^{-1}\sum_{i=1}^{n}\rho\left\{ Y_{it}(\beta_{t})-\mu(A_{i},X_{i}\mid\gamma)\right\} $,
where $\rho(x)=0.5x^{2}\mathbb{I}(|x|\ensuremath{<l})+\left\{ l|x|-0.5l^{2}\right\} \mathbb{I}(|x|\ensuremath{\geq l})$.
Then the estimator based on the Huber loss is
\begin{align}
(\hat{\gamma}^{(0)},\hat{\zeta}) & =\text{argmin}L_{n}(\gamma^{(0)},\zeta)\nonumber \\
 & =\text{argmin}\left\{ L_{n}(\gamma^{(0)},\zeta)-L_{n}(\gamma_{0}^{(0)},\zeta)\right\} +\left\{ L_{n}(\gamma_{0}^{(0)},\zeta)-L_{n}(\gamma_{0}^{(0)},\zeta')\right\} ,\label{eq:loss_part}
\end{align}
where $\zeta'$ is a fixed value. We examine the two terms in the
objective function \eqref{eq:loss_part} separately. 

For the first term in the function \eqref{eq:loss_part}, $L_{n}(\gamma^{(0)},\zeta)-L_{n}(\gamma_{0}^{(0)},\zeta)$
\begin{align*}
 & =\frac{1}{n}\sum_{i=1}^{n}\rho\left[\tilde{h}(X)+\varepsilon_{i}-\tilde{h}(X;\zeta)-\left(A-\pi\right)\left\{ g(X;\gamma^{(0)})-g(X;\gamma_{0}^{(0)})\right\} \right]\\
 & \qquad-\rho\left[\tilde{h}(X)+\varepsilon_{i}-\tilde{h}(X;\zeta)-\left(A-\pi\right)\left\{ g(X;\gamma^{(0)})-g(X;\gamma_{0}^{(0)})\right\} \right]\\
: & =\frac{1}{n}\sum_{i=1}^{n}d_{i1}.
\end{align*}
The regularity conditions C4 and C5 allow us to apply the weak law
of large number (WLLN) to $L_{n}(\gamma^{(0)},\zeta)-L_{n}(\gamma_{0}^{(0)},\zeta)$,
since 
\begin{align*}
|d_{i1}| & \leq\big|\frac{1}{2}\left(K_{1}-K_{2}\right)^{2}-l|K_{1}|+\frac{1}{2}l^{2}\big|\\
 & =\frac{1}{2}\left(K_{1}-l\right)^{2}+|K_{2}||K_{2}-2K_{1}|
\end{align*}
has a finite expectation. Thus, by WLLN, $L_{n}(\gamma^{(0)},\zeta)-L_{n}(\gamma_{0}^{(0)},\zeta)\xrightarrow{\pr}\E(D_{1})=G_{1}(\gamma^{(0)},\zeta)$,
where $D_{1}=\rho(K_{1}-K_{2})-\rho(K_{1})$. We claim that $G_{1}(\gamma^{(0)},\zeta)\geq0$
and reaches $0$ if and only if $\gamma^{(0)}=\gamma_{0}^{(0)}$.

First, note that $G_{1}(\gamma_{0}^{(0)},\zeta)=\rho(K_{1})-\rho(K_{1})=0$.
We proceed to prove that $G_{1}(\gamma^{(0)},\zeta)>0$ for $\forall\gamma^{(0)}\neq\gamma_{0}^{(0)}$
and consider the following four cases:
\begin{enumerate}
\item If $K_{1}>l$, we have 
\[
D_{1}=\rho(K_{1}-K_{2})-\rho(K_{1})\geq l(K_{1}-K_{2})-\frac{1}{2}l^{2}-\left(lK_{1}-\frac{1}{2}l^{2}\right)=-lK_{2}.
\]
\item If $K_{1}<-l$, repeat the step in 1 and we can get $D_{1}\geq lK_{2}$.
\item If $K_{1}\in[-l,l]$ and $K_{1}-K_{2}\in[-l,l]$, then 
\[
D_{1}=\frac{1}{2}(K_{1}-K_{2})^{2}-\frac{1}{2}K_{1}^{2}=-K_{1}K_{2}+\frac{1}{2}K_{2}^{2}.
\]
\item If $K_{1}\in[-l,l]$ and $K_{1}-K_{2}\notin[-l,l]$, then 
\begin{align*}
D_{1} & =l|K_{1}-K_{2}|-\frac{1}{2}l^{2}-\frac{1}{2}K_{1}^{2}\\
 & =\frac{1}{2}(K_{1}-K_{2})^{2}-\frac{1}{2}(|K_{1}-K_{2}|-l)^{2}-\frac{1}{2}K_{1}^{2}\\
 & \geq\frac{1}{2}(K_{1}-K_{2})^{2}-\frac{1}{2}K_{2}^{2}-\frac{1}{2}K_{1}^{2}=-K_{1}K_{2}.
\end{align*}
The inequality holds since by triangle inequality, $0<|K_{1}-K_{2}|-l\leq K_{1}+|K_{2}|-l\leq|K_{2}|$. 
\end{enumerate}
Incorporating the four cases together and taking expectations, we
have
\begin{align*}
G_{1}(\gamma^{(0)},\zeta) & \geq\E\left\{ -lK_{2}\mathbb{I}(K_{1}>l)\right\} +\E\left\{ lK_{2}\mathbb{I}(K_{1}<-l)\right\} \\
 & \qquad+\E\left\{ \left(-K_{1}K_{2}+\frac{1}{2}K_{2}^{2}\right)\mathbb{I}(K_{1}\in[-l,l])\mathbb{I}\left(K_{1}-K_{2}\in[-l,l]\right)\right\} \\
 & \qquad+\E\left\{ -K_{1}K_{2}\mathbb{I}(K_{1}\in[-l,l])\mathbb{I}\left(K_{1}-K_{2}\notin[-l,l]\right)\right\} \\
 & \geq\E\left\{ -lK_{2}\mathbb{I}(K_{1}>l)\right\} +\E\left\{ lK_{2}\mathbb{I}(K_{1}<-l)\right\} +\E\left\{ -K_{1}K_{2}\mathbb{I}(K_{1}\in[-l,l])\right\} \\
 & \qquad+\E\left\{ \frac{1}{2}K_{2}^{2}\mathbb{I}(K_{1}\in[-l,l])\mathbb{I}\left(K_{1}-K_{2}\in[-l,l]\right)\right\} .
\end{align*}
Note that by the regularity condition C3, $\pr(A)=\pi$, and $A\indep X$
(randomized trial), we have
\begin{align*}
\E\left\{ -lK_{2}\mathbb{I}(K_{1}>l)\right\}  & =-l\E\left[\E(K_{2}\mid X)\E\left\{ \mathbb{I}(K_{1}>l)\mid X\right\} \right]\\
 & =-l\E\left(\E\left[\left(A-\pi\right)\left\{ g(X;\gamma^{(0)})-g(X;\gamma_{0}^{(0)})\right\} \mid X\right]\E\left\{ \mathbb{I}(K_{1}>l)\mid X\right\} \right)\\
 & =-l\E\left[\E\left(A_{i}-\pi\right)\left\{ g(X;\gamma^{(0)})-g(X;\gamma_{0}^{(0)})\right\} \E\left\{ \mathbb{I}(K_{1}>l)\mid X\right\} \right]=0.
\end{align*}
Similarly, we have $\E\left\{ lK_{2}\mathbb{I}(K_{1}<-l)\right\} =\E\left\{ -K_{1}K_{2}\mathbb{I}(K_{1}\in[-l,l])\right\} =0$.
Therefore, 
\[
G_{1}(\gamma^{(0)},\zeta)\geq\E\left\{ \frac{1}{2}K_{2}^{2}\mathbb{I}(K_{1}\in[-l,l])\mathbb{I}\left(K_{1}-K_{2}\in[-l,l]\right)\right\} .
\]
By the regularity conditions C6 and C7, we know that $G_{1}(\gamma^{(0)},\zeta)>0$
for $\forall\gamma^{(0)}\neq\gamma_{0}^{(0)}$. 

For the second term in the function \eqref{eq:loss_part}, denote
$L_{n}(\gamma_{0}^{(0)},\zeta)-L_{n}(\gamma_{0}^{(0)},\zeta')=n^{-1}\sum_{i=1}^{n}d_{i2}$.
By the regularity condition C4, WLLN is applied, and we have $L_{n}(\gamma_{0}^{(0)},\zeta)-L_{n}(\gamma_{0}^{(0)},\zeta')\xrightarrow{\pr}\E(D_{2})=G_{2}(\zeta)$.

The results for the first term combined with the regularity condition
C8 implies that $(\gamma_{0}^{(0)},\zeta^{*})$ is the unique minimizer
of $G_{1}(\gamma^{(0)},\zeta)+G_{2}(\zeta)$. Since the Huber loss
function is strongly convex, by the argmax continuous mapping theorem,
we have $\hat{\gamma}^{(0)}\xrightarrow{\pr}\gamma_{0}^{(0)}$. By
continuous mapping theorem, $\hat{\tau}=n^{-1}\sum_{i=1}^{n}g(X_{i};\hat{\gamma}^{(0)})\xrightarrow{\pr}\E\left\{ g(X_{i};\gamma^{(0)})\right\} =\tau_{0}$.

\paragraph{(ii) For the absolute loss,}

denote $L_{a,n}(\gamma^{(0)},\zeta)=n^{-1}\sum_{i=1}^{n}\rho_{a}\left\{ Y_{i}-\mu(A,X\mid\gamma)\right\} $,
where $\rho_{a}(x)=|x|$ is the absolute loss. Then the estimator
based on $\rho_{a}(x)$ is
\begin{align}
(\hat{\gamma}^{(0)},\hat{\zeta}) & =\text{argmin}L_{a,n}(\gamma^{(0)},\zeta)\nonumber \\
 & =\text{argmin}\left\{ L_{a,n}(\gamma^{(0)},\zeta)-L_{a,n}(\gamma_{0}^{(0)},\zeta)\right\} +\left\{ L_{a,n}(\gamma_{0}^{(0)},\zeta)-L_{a,n}(\gamma_{0}^{(0)},\zeta')\right\} ,\label{eq:absloss_part-1}
\end{align}
where $\zeta'$ is a fixed value. We again examine the two terms in
the objective function \eqref{eq:absloss_part-1} separately. 

For the first term in the function \eqref{eq:absloss_part-1}, $L_{a,n}(\gamma^{(0)},\zeta)-L_{a,n}(\gamma_{0}^{(0)},\zeta)$
\begin{align*}
 & =\frac{1}{n}\sum_{i=1}^{n}\bigg(\rho_{a}\left[\tilde{h}(X)+\varepsilon_{i}-\tilde{h}(X;\zeta)-\left(A-\pi\right)\left\{ g(X;\gamma^{(0)})-g(X;\gamma_{0}^{(0)})\right\} \right]\\
 & \qquad-\rho_{a}\left[\tilde{h}(X)+\varepsilon_{i}-\tilde{h}(X;\zeta)-\left(A-\pi\right)\left\{ g(X;\gamma^{(0)})-g(X;\gamma_{0}^{(0)})\right\} \right]\bigg)\\
: & =\frac{1}{n}\sum_{i=1}^{n}d_{i1}.
\end{align*}
The regularity conditions C4 and C5 allow us to apply WLLN to $L_{a,n}(\gamma^{(0)},\zeta)-L_{a,n}(\gamma_{0}^{(0)},\zeta)$,
since $|d_{i1}|\leq\big|K_{1}-K_{2}|+|K_{1}|\leq2|K_{1}|+|K_{2}|$
has a finite expectation. Thus, by WLLN, $L_{a,n}(\gamma^{(0)},\zeta)-L_{a,n}(\gamma_{0}^{(0)},\zeta)\xrightarrow{\pr}\E(D_{1})=G_{1}(\gamma^{(0)},\zeta)$,
where $D_{1}=\rho_{a}(K_{1}-K_{2})-\rho_{a}(K_{1})$. 

First, note that $G_{1}(\gamma_{0}^{(0)},\zeta)=\rho_{a}(K_{1})-\rho_{a}(K_{1})=0$.
We proceed to prove that $G_{1}(\gamma^{(0)},\zeta)>0$ for $\forall\gamma^{(0)}\neq\gamma_{0}^{(0)}$
and consider the following two cases:
\begin{enumerate}
\item If $K_{1}\geq0$, we have $D_{1}=|K_{1}-K_{2}|-|K_{1}|\geq K_{1}-K_{2}-K_{1}=-K_{2}.$
\item If $K_{1}<0$, we have $D_{1}=|K_{1}-K_{2}|-|K_{1}|\geq K_{2}-K_{1}+K_{1}=K_{2}.$
\end{enumerate}
Incorporating the four cases together and taking expectations, we
have $G_{1}(\gamma^{(0)},\zeta)\geq\E\big\{ -K_{2}\mathbb{I}(K_{1}\geq0)\big\} +\E\left\{ K_{2}\mathbb{I}(K_{1}<0)\right\} $.
Follow the same proof, we have $\E\left\{ -K_{2}\mathbb{I}(K_{1}\geq0)\right\} =\E\left\{ K_{2}\mathbb{I}(K_{1}<0)\right\} =0$.
By the regularity conditions C6 and C7, we know that $G_{1}(\gamma^{(0)},\eta)>0$
for $\forall\gamma^{(0)}\neq\gamma_{0}^{(0)}$. The remaining proof
follows similar steps the proof for the Huber loss. 

\paragraph*{(iii) For the $\varepsilon$-insensitive loss,}

denote $L_{\varepsilon,n}(\gamma^{(0)},\zeta)=n^{-1}\sum_{i=1}^{n}\mathcal{L}_{\varepsilon}\left\{ Y_{i}-\mu(A,X\mid\gamma)\right\} $,
where $\mathcal{L}_{\varepsilon}(x)=\max\left\{ |x|-\varepsilon,0\right\} $
is the $\varepsilon$-insensitive loss. Then the estimator based on
$\mathcal{L}_{\varepsilon}(x)$ is
\begin{align}
(\hat{\gamma}^{(0)},\hat{\zeta}) & =\text{argmin}L_{\varepsilon,n}(\gamma^{(0)},\zeta)\nonumber \\
 & =\text{argmin}\left\{ L_{\varepsilon,n}(\gamma^{(0)},\zeta)-L_{\varepsilon,n}(\gamma_{0}^{(0)},\zeta)\right\} +\left\{ L_{\varepsilon,n}(\gamma_{0}^{(0)},\zeta)-L_{\varepsilon,n}(\gamma_{0}^{(0)},\zeta')\right\} ,\label{eq:epsloss_part-1}
\end{align}
where $\zeta'$ is a fixed value. We again examine the two terms in
the objective function \eqref{eq:epsloss_part-1} separately. 

For the first term in the function \eqref{eq:epsloss_part-1}, $L_{\varepsilon,n}(\gamma^{(0)},\zeta)-L_{\varepsilon,n}(\gamma_{0}^{(0)},\zeta)$
\begin{align*}
 & =\frac{1}{n}\sum_{i=1}^{n}\bigg(\mathcal{L}_{\varepsilon}\left[\tilde{h}(X)+\varepsilon_{i}-\tilde{h}(X;\zeta)-\left(A-\pi\right)\left\{ g(X;\gamma^{(0)})-g(X;\gamma_{0}^{(0)})\right\} \right]\\
 & \qquad-\mathcal{L}_{\varepsilon}\left[\tilde{h}(X)+\varepsilon_{i}-\tilde{h}(X;\zeta)-\left(A-\pi\right)\left\{ g(X;\gamma^{(0)})-g(X;\gamma_{0}^{(0)})\right\} \right]\bigg)\\
: & =\frac{1}{n}\sum_{i=1}^{n}d_{i1}.
\end{align*}
The regularity conditions C4 and C5 allow us to apply WLLN to $L_{\varepsilon,n}(\gamma^{(0)},\zeta)-L_{\varepsilon,n}(\gamma_{0}^{(0)},\zeta)$,
since $|d_{i1}|\leq\big|K_{1}-K_{2}|+\varepsilon+|K_{1}|+\varepsilon\leq2|K_{1}|+|K_{2}|+2\varepsilon$
has a finite expectation. Thus, by WLLN, $L_{\varepsilon,n}(\gamma^{(0)},\zeta)-L_{\varepsilon,n}(\gamma_{0}^{(0)},\zeta)\xrightarrow{\pr}\E(D_{1})=G_{1}(\gamma^{(0)},\zeta)$,
where $D_{1}=\mathcal{L}_{\varepsilon}(K_{1}-K_{2})-\mathcal{L}_{\varepsilon}(K_{1})$. 

First, note that $G_{1}(\gamma_{0}^{(0)},\zeta)=\mathcal{L}_{\varepsilon}(K_{1})-\mathcal{L}_{\varepsilon}(K_{1})=0$.
We proceed to prove that $G_{1}(\gamma^{(0)},\zeta)>0$ for $\forall\gamma^{(0)}\neq\gamma_{0}^{(0)}$
and consider the following two cases:
\begin{enumerate}
\item If $K_{1}\geq\varepsilon$, we have $D_{1}=\max\left(|K_{1}-K_{2}|-\varepsilon,0\right)-K_{1}+\varepsilon\geq|K_{1}-K_{2}|-|K_{1}|\geq K_{1}-K_{2}-K_{1}=-K_{2}$.
\item If $K_{1}\leq-\varepsilon$, we have $D_{1}=\max\left(|K_{1}-K_{2}|-\varepsilon,0\right)+K_{1}+\varepsilon\geq|K_{1}-K_{2}|-|K_{1}|\geq K_{2}-K_{1}+K_{1}=K_{2}$.
\item If $|K_{1}|<\varepsilon$, we have $D_{1}=\max\left(|K_{1}-K_{2}|-\varepsilon,0\right)-0\geq0$.
\end{enumerate}
Incorporating the four cases together and taking expectations, we
have $G_{1}(\gamma^{(0)},\zeta)\geq\E\big\{ -K_{2}\mathbb{I}(K_{1}\geq\varepsilon)\big\} +\E\left\{ K_{2}\mathbb{I}(K_{1}\leq-\varepsilon)\right\} $.
Follow the same proof, we have $\E\left\{ -K_{2}\mathbb{I}(K_{1}\geq\varepsilon)\right\} =\E\left\{ K_{2}\mathbb{I}(K_{1}\leq\varepsilon)\right\} =0$.
By the regularity conditions C6 and C7, we know that $G_{1}(\gamma^{(0)},\zeta)>0$
for $\forall\gamma^{(0)}\neq\gamma_{0}^{(0)}$. The remaining proof
follows similar steps the proof for the Huber loss. \end{proof}

\subsection{Proof of Theorem 2 \label{subsec:supp_norm}}

To explore the asymptotic normality of the estimator $\hat{\tau},$
we first focus on the asymptotic normality of $\hat{\alpha}_{s-1}^{w}$
for $s=1,\cdots,t$. Denote $\varphi(H_{s},\alpha_{s-1}):=(1-A)R_{s}w(H_{s-1};\rho_{s-1})\psi(Y_{s}-H_{s-1}^{\T}\alpha_{s-1})H_{s-1}$,
where $\psi(x)=\partial\rho(x)/\partial x$ is the derivative of the
robust loss function. Therefore, $\hat{\alpha}_{s-1}^{w}$ is the
solution to the estimating equations $\sum_{i=1}^{n}\varphi(H_{is},\alpha_{s-1})=0$.

\begin{lemma}\label{lemma:alpha_norm}

Assume the regularity conditions C1 and C2 and the following conditions:

\begin{enumerate}

\item[C9.] The partial derivative $\E\left\{ \psi(Y_{s}-H_{s-1}^{\T}\alpha_{s-1})\mid H_{s-1}\right\} $
with respect to $\alpha_{s-1}$ exists and is continuous around $\alpha_{s-1,0}$
almost everywhere. The second derivative of $\E\left\{ \psi(Y_{s}-H_{s-1}^{\T}\alpha_{s-1})\mid H_{s-1}\right\} $
with respect to $\alpha_{s-1}$ is continuous and dominated by some
integrable functions;

\item[C10.] The partial derivative of $\E\left\{ \varphi(H_{s},\alpha_{s-1})\mid H_{s-1}\right\} $
with respect to $\alpha_{s-1}$ is nonsingular. 

\item[C11.] The variance $\mathbb{V}\{q(H_{s},\alpha_{s-1,0})\}$
is finite, where 
\[
q(H_{s},\alpha_{s-1,0})=\left[-\frac{\partial\E\left\{ \varphi(H_{is},\alpha_{s-1,0})H_{is-1}^{\T}\mid H_{is-1}\right\} }{\partial\alpha_{s-1}^{\T}}\right]^{-1}\varphi(H_{s},\alpha_{s-1,0}).
\]

\end{enumerate} Then, for $s=1,\cdots,t$, as the sample size $n\rightarrow\infty$,
\[
\sqrt{n}(\hat{\alpha}_{s-1}^{w}-\alpha_{s-1,0})\xrightarrow{d}\mathcal{N}\Big(0,\mathbb{V}\left\{ q(H_{s},\alpha_{s-1,0})\right\} \Big).
\]

\end{lemma}

\begin{proof} Consider a Taylor expansion of the function $R_{s}\varphi(H_{s},\hat{\alpha}_{s-1}^{w})$
with respect to $\hat{\alpha}_{s-1}^{w}$ around $\alpha_{s-1,0}$
for $s=1,\cdots,t$, under the regularity conditions C1, C9 and C10,
we have the linearization form of $\hat{\alpha}_{s-1}^{w}$ as
\begin{align*}
\hat{\alpha}_{s-1}^{w}-\alpha_{s-1,0} & =\frac{1}{n}\sum_{i=1}^{n}\left[-\frac{\partial\E\left\{ \varphi(H_{is},\alpha_{s-1,0})H_{is-1}^{\T}\mid H_{is-1}\right\} }{\partial\alpha_{s-1}^{\T}}\right]^{-1}\varphi(H_{is},\alpha_{s-1,0})+o_{\pr}(n^{-1/2})\\
 & =\frac{1}{n}\sum_{i=1}^{n}q(H_{is},\alpha_{s-1,0})+o_{\pr}(n^{-1/2}).
\end{align*}
Under the regularity condition C11, we apply the central limit theorem
and get the asymptotic distribution of $\hat{\alpha}_{s-1}^{w}$.
\end{proof}

For simplicity of the notations, denote $\alpha_{0}=(\alpha_{0,0}^{\T},\cdots,\alpha_{t-1,0}^{\T})^{\T}$
as the true model parameters from $t$ sequential regression models.
Based on Lemmas \ref{lemma:beta-alpha} and \ref{lemma:alpha_norm},
we can further obtain the asymptotic normality of $\hat{\beta}_{t,s}$
for $s=1,\cdots,t$ in the following lemma.

\begin{lemma}\label{lemma:beta_norm}

Assume the regularity conditions C1--C11 and the following conditions:

\begin{enumerate}

\item[C12.] The variance $\mathbb{V}\{U_{t,s-1,i}(\alpha_{0})\}$
is finite, where $U_{t,s-1,i}(\alpha)$ is the linearization form
produced by $\hat{\beta}_{t,s-1}$, i.e., $\hat{\beta}_{t,s-1}-\beta_{t,s-1}=n^{-1}\sum_{i=1}^{n}U_{t,s-1,i}(\alpha_{0})+o_{\pr}(n^{-1/2})$.
Specifically, $U_{t,s,i}(\alpha_{0})=\left(\I_{p+s-2},\alpha_{s-1,0}\right)U_{t,s+1,i}(\alpha_{0})+\left(\mathbf{0}_{p+s-2}^{\T},1\right)\beta_{t,s}q(H_{is},\alpha_{s-1,0})$
and $U_{t,t-1,i}(\alpha_{0})=q(H_{it},\alpha_{t-1,0})$.

\end{enumerate}

Then, as the sample size $n\rightarrow\infty$, for $s=1,\cdots,t$,
\[
\sqrt{n}(\hat{\beta}_{t,s-1}-\beta_{t,s-1})\xrightarrow{d}\mathcal{N}\Big(0,\mathbb{V}\left\{ U_{t,s-1,i}(\alpha_{0})\right\} \Big).
\]

\end{lemma}

\begin{proof} Lemma \ref{lemma:beta-alpha} indicates that $\hat{\beta}_{t,t-1}=\hat{\alpha}_{t-1}^{w}$,
thus $\hat{\beta}_{t,t-1}$ shares the same linearization form as
$\hat{\alpha}_{t-1}^{w}$, i.e., 
\begin{align*}
\hat{\beta}_{t,t-1}-\beta_{t,t-1} & =\frac{1}{n}\sum_{i=1}^{n}q(H_{it},\alpha_{t-1,0})+o_{\pr}(n^{-1/2})\\
 & =\frac{1}{n}\sum_{i=1}^{n}U_{t,t-1,i}(\alpha_{0})+o_{\pr}(n^{-1/2}).
\end{align*}

For the individuals who drop out at visit $t-1$, the corresponding
imputation parameter estimate $\hat{\beta}_{t,t-2}$ can be expressed
as
\begin{align}
\hat{\beta}_{t,t-2} & =\left(\I_{p+t-3},\hat{\alpha}_{t-2,0}^{w}\right)\hat{\alpha}_{t-1}^{w}\nonumber \\
 & =\left(\I_{p+t-3},\hat{\alpha}_{t-2,0}^{w}\right)\hat{\beta}_{t,t-1}\nonumber \\
 & =\left(\I_{p+t-3},\mathbf{0}_{p+t-3}^{\T}\right)\hat{\beta}_{t,t-1}+\hat{\alpha}_{t-2}^{w}\left(\mathbf{0}_{p+t-3}^{\T},1\right)\hat{\beta}_{t,t-1}.\label{eq:beta1}
\end{align}
The linearization form of the first term in formula \eqref{eq:beta1}
can be obtained directly via delta-method as
\[
\left(\I_{p+t-3},\mathbf{0}_{p+t-3}^{\T}\right)\hat{\beta}_{t,t-1}-\left(\I_{p+t-3},\mathbf{0}_{p+t-3}^{\T}\right)\beta_{t,t-1}=\frac{1}{n}\sum_{i=1}^{n}\left(\I_{p+t-3},\mathbf{0}_{p+t-3}^{\T}\right)q(H_{it},\alpha_{t-1,0})+o_{\pr}(n^{-1/2}).
\]
For the second term, let $g_{t-2}(\alpha_{t-2},\beta_{t,t-1}):=\alpha_{t-2}\left(\mathbf{0}_{p+t-3}^{\T},1\right)\beta_{t,t-1}$.
Then we have $\nabla g_{t-2}\Big|_{(\alpha_{t-2,0},\beta_{t,t-1})}=\left\{ \left(\mathbf{0}_{p+t-3}^{\T},1\right)\beta_{t,t-1},\alpha_{t-2,0}\left(\mathbf{0}_{p+t-3}^{\T},1\right)\right\} ^{\T}$.
Under the regularity condition C9 by Theorem 5.27 in \citet{boos2013essential},
we have
\begin{align*}
g_{t-2}(\hat{\alpha}_{t-2}^{w},\hat{\beta}_{t,t-1})-g_{t-2}(\alpha_{t-2,0},\beta_{t,t-1}) & =\frac{1}{n}\sum_{i=1}^{n}\nabla g_{t-2}^{\T}\Big|_{(\alpha_{t-2,0},\beta_{t,t-1})}\begin{pmatrix}q(H_{it-1},\alpha_{t-2,0})\\
U_{t,t-1,i}(\alpha_{0})
\end{pmatrix}+o_{\pr}(n^{-1/2})\\
 & =\frac{1}{n}\sum_{i=1}^{n}\left(\mathbf{0}_{p+t-3}^{\T},1\right)\beta_{t,t-1}q(H_{it-1},\alpha_{t-2,0})\\
 & \qquad\qquad+\alpha_{t-2,0}\left(\mathbf{0}_{p+t-3}^{\T},1\right)U_{t,t-1,i}(\alpha_{0})+o_{\pr}(n^{-1/2}).
\end{align*}
Combine the two terms together, we have 
\begin{align*}
\hat{\beta}_{t,t-2}-\beta_{t,t-2} & =\frac{1}{n}\sum_{i=1}^{n}\left(\mathbf{0}_{p+t-3}^{\T},1\right)\beta_{t,t-1}q(H_{it-1},\alpha_{t-2,0})\\
 & \qquad\qquad+\left\{ \left(\I_{p+t-3},\mathbf{0}_{p+t-3}^{\T}\right)+\alpha_{t-2,0}\left(\mathbf{0}_{p+t-3}^{\T},1\right)\right\} U_{t,t-1,i}(\alpha_{0})+o_{\pr}(n^{-1/2})\\
 & =\frac{1}{n}\sum_{i=1}^{n}\left(\I_{p+t-3},\alpha_{t-2,0}\right)U_{t,t-1,i}(\alpha_{0})+\left(\mathbf{0}_{p+t-3}^{\T},1\right)\beta_{t,t-1}q(H_{it-1},\alpha_{t-2,0})+o_{\pr}(n^{-1/2})\\
 & =\frac{1}{n}\sum_{i=1}^{n}U_{t,t-2,i}(\alpha_{0})+o_{\pr}(n^{-1/2}),
\end{align*}
which matches the result in the lemma when $s=t-2$.

We then prove the lemma by induction. Suppose the result holds for
the individual who drops out at visit $s+1$, i.e., 
\begin{align*}
\hat{\beta}_{t,s}-\beta_{t,s} & =\frac{1}{n}\sum_{i=1}^{n}U_{t,s,i}(\alpha_{0})+o_{\pr}(n^{-1/2})\\
 & =\frac{1}{n}\sum_{i=1}^{n}\left(\I_{p+s-1},\alpha_{s,0}\right)U_{t,s,i}(\alpha_{0})+\left(\mathbf{0}_{p+s-1}^{\T},1\right)\beta_{t,s+1}q(H_{is+1},\alpha_{s,0})+o_{\pr}(n^{-1/2}).
\end{align*}
Then for individuals in group $a$ who drop out at visit $s$, the
corresponding imputation parameter estimate $\hat{\beta}_{t,s-1}$
can be expressed as
\begin{align}
\hat{\beta}_{t,s-1} & =(\I_{p+s-2},\hat{\alpha}_{s-1,0}^{w})\hat{\beta}_{t,s}\nonumber \\
 & =\left(\I_{p+s-2},\mathbf{0}_{p+s-2}^{\T}\right)\hat{\beta}_{t,s}+\hat{\alpha}_{s-1,0}^{w}\left(\mathbf{0}_{p+s-2}^{\T},1\right)\hat{\beta}_{t,s}.\label{eq:beta2}
\end{align}
Similarly, the linearization form of the first term in formula \eqref{eq:beta2}
can be obtained directly via delta-method as
\[
\left(\I_{p+s-2},\mathbf{0}_{p+s-2}^{\T}\right)\hat{\beta}_{t,s}-\left(\I_{p+s-2},\mathbf{0}_{p+s-2}^{\T}\right)\beta_{t,s}=\frac{1}{n}\sum_{i=1}^{n}\left(\I_{p+s-2},\mathbf{0}_{p+s-2}^{\T}\right)U_{t,s,i}(\alpha_{0})+o_{\pr}(n^{-1/2}).
\]
For the second term, let $g_{s-1}(\alpha_{s-1},\beta_{t,s}):=\alpha_{s-1}\left(\mathbf{0}_{p+s-2}^{\T},1\right)\beta_{t,s}$.
Then we have $\nabla g_{s-1}\Big|_{(\alpha_{s-1,0},\beta_{t,s})}=\left\{ \left(\mathbf{0}_{p+s-2}^{\T},1\right)\beta_{t,s},\alpha_{s-1,0}\left(\mathbf{0}_{p+s-2}^{\T},1\right)\right\} ^{\T}$.
Under the regularity condition C9, we have
\begin{align*}
g_{s-1}(\hat{\alpha}_{s-1}^{w},\hat{\beta}_{t,s})-g_{t-2}(\alpha_{s-1,0},\beta_{t,s}) & =\frac{1}{n}\sum_{i=1}^{n}\nabla g_{s-1}^{\T}\Big|_{(\alpha_{s-1,0},\beta_{t,s})}\begin{pmatrix}q(H_{is},\alpha_{s-1,0})\\
U_{t,s,i}(\alpha_{0})
\end{pmatrix}+o_{\pr}(n^{-1/2})\\
 & =\frac{1}{n}\sum_{i=1}^{n}\left(\mathbf{0}_{p+s-2}^{\T},1\right)\beta_{t,s}q(H_{is},\alpha_{s-1,0})\\
 & \qquad\qquad+\alpha_{s-1,0}\left(\mathbf{0}_{p+s-2}^{\T},1\right)U_{t,s,i}(\alpha_{0})+o_{\pr}(n^{-1/2}).
\end{align*}
Combine the two terms, the linearization form of $\hat{\beta}_{t,s-1}$
is
\begin{align*}
\hat{\beta}_{t,s-1}-\beta_{t,s-1} & =\frac{1}{n}\sum_{i=1}^{n}\left(\mathbf{0}_{p+s-2}^{\T},1\right)\beta_{t,s}q(H_{is},\alpha_{s-1,0})\\
 & \qquad\qquad+\left\{ \left(\I_{p+s-2},\mathbf{0}_{p+s-2}^{\T}\right)+\alpha_{s-1,0}\left(\mathbf{0}_{p+s-2}^{\T},1\right)\right\} U_{t,s,i}(\alpha_{0})+o_{\pr}(n^{-1/2})\\
 & =\frac{1}{n}\sum_{i=1}^{n}\left(\I_{p+s-2},\alpha_{s-1,0}\right)U_{t,s,i}(\alpha_{0})+\left(\mathbf{0}_{p+s-2}^{\T},1\right)\beta_{t,s}q(H_{is},\alpha_{s-1,0})+o_{\pr}(n^{-1/2})\\
 & =\frac{1}{n}\sum_{i=1}^{n}U_{t,s-1,i}(\alpha_{0})+o_{\pr}(n^{-1/2}).
\end{align*}
Apply the central limit theorem based on the regularity condition
C12 and we complete the proof. \end{proof}

We restate Theorem 2 in the main text below with technical details. 

\begin{thm}

Under the regularity conditions C1--C12, and assume the following
regularity conditions:

\begin{enumerate}

\item[C13.] The partial derivatives $\varphi_{a}\left\{ Z^{*}(\beta_{t}),\gamma\right\} $
with respect to $\gamma$ and $\beta_{t}$ exist and are continuous
around $\gamma_{0}$ and $\beta_{t}$ almost everywhere. The second
derivatives of $\varphi_{a}\left\{ Z^{*}(\beta_{t}),\gamma\right\} $
with respect to $\gamma_{0}$ and $\beta_{t}$ are continuous and
dominated by some integrable functions.

\item[C14.] The partial derivative of $\E\left[\varphi_{a}\left\{ Z^{*}(\beta_{t}),\gamma\right\} \right]$
with respect to $\gamma$ at $\gamma=\gamma_{0}$, i.e., $D_{\varphi} = \partial\E\Big[\varphi_{a}\big\{ Z^{*}(\beta_{t}),\allowbreak \gamma_{0}\big\}\Big] /\partial\gamma^{\T}$,
is nonsingular. 

\item[C15.] The partial derivative $g(X,\gamma^{(0)})$ with respect
to $\gamma^{(0)}$ exists and is continuous around $\gamma_{0}^{(0)}$
almost everywhere. The second derivative of $g(X,\gamma^{(0)})$ with
respect to $\gamma_{0}^{(0)}$ is continuous and dominated by some
integrable functions.

\item[C16.] The variance $\mathbb{V}\left\{ V_{\tau,i}(\alpha_{0},\gamma_{0})\right\} $
is finite, where $V_{\tau,i}(\alpha_{0},\gamma_{0})=\left\{ \partial g(X_{i};\gamma_{0}^{(0)})/\partial\gamma^{\T}\right\} c^{\T}V_{\gamma,i}(\alpha_{0},\gamma_{0})$,
\begin{align*}
V_{\gamma,i}(\alpha_{0},\gamma_{0}) & =D_{\varphi}^{-1}\bigg[\varphi_{a}\left\{ Z_{i}^{*}(\beta_{t}),\gamma_{0}\right\} +\sum_{s=1}^{t}\E\left\{ R_{s-1}(1-R_{s})\frac{\partial\mu(A,X\mid\gamma_{0})}{\partial\gamma^{\T}}\frac{\partial\psi(e)}{\partial e}H_{s-1}^{\T}\right\} U_{t,s-1,i}(\alpha_{0})\bigg],
\end{align*}
$e_{i}=Y_{it}^{*}(\beta_{t})-\mu(A,X\mid\gamma_{0})$, and $c^{\T}=(\mathbf{I}_{d_{0}},\mathbf{0}_{d_{0}\times d_{1}})$.
Here, $\I_{d_{0}}$ is a $(d_{0}\times d_{0})$-dimensional identity
matrix, $\mathbf{0}_{d_{0}\times d_{1}}$ is a $(d_{0}\times d_{1})$-dimensional
zero matrix. 

\end{enumerate} Then, as the sample size $n\rightarrow\infty$,
\[
\sqrt{n}(\hat{\tau}-\tau_{0})\xrightarrow{d}\mathcal{N}\Big(0,\mathbb{V}\left\{ V_{\tau,i}(\alpha_{0},\gamma_{0})\right\} \Big).
\]

\end{thm}

\begin{proof} Consider a Taylor expansion of the function $\sum_{i=1}^{n}\varphi_{a}\left\{ Z_{i}^{*}(\hat{\beta}_{t}),\gamma\right\} $
with respect to $\gamma$ around $\gamma_{0}$, under the regularity
conditions C13 and C14, we have the linearization form of $\hat{\gamma}$
as
\begin{align*}
\hat{\gamma}-\gamma_{0} & =\frac{1}{n}\sum_{i=1}^{n}\left[-\frac{1}{n}\sum_{i=1}^{n}\frac{\partial\varphi_{a}\left\{ Z_{i}^{*}(\hat{\beta}_{t}),\gamma_{0}\right\} }{\partial\gamma^{\T}}\right]^{-1}\varphi_{a}\left\{ Z_{i}^{*}(\hat{\beta}_{t}),\gamma_{0}\right\} +o_{\pr}(n^{-1/2})\\
 & =\left(-\frac{\partial\E\left[\varphi_{a}\left\{ Z(\beta_{t}),\gamma_{0}\right\} \right]}{\partial\gamma^{\T}}\right)^{-1}\frac{1}{n}\sum_{i=1}^{n}\varphi_{a}\left\{ Z_{i}^{*}(\hat{\beta}_{t}),\gamma_{0}\right\} +o_{\pr}(n^{-1/2})\\
 & =D_{\varphi}^{-1}\frac{1}{n}\sum_{i=1}^{n}\varphi_{a}\left\{ Z_{i}^{*}(\hat{\beta}_{t}),\gamma_{0}\right\} +o_{\pr}(n^{-1/2}).
\end{align*}
Therefore, 
\begin{align}
\sqrt{n}(\hat{\gamma}-\gamma_{0}) & =D_{\varphi}^{-1}n^{-1/2}\sum_{i=1}^{n}\varphi_{a}\left\{ Z_{i}^{*}(\hat{\beta}_{t}),\gamma_{0}\right\} +o_{\pr}(1)\nonumber \\
 & =D_{\varphi}^{-1}\Big[n^{-1/2}\sum_{i=1}^{n}\varphi_{a}\left\{ Z_{i}^{*}(\beta_{t}),\gamma_{0}\right\} \nonumber \\
 & \qquad\quad+n^{-1/2}\sum_{i=1}^{n}\varphi_{a}\left\{ Z_{i}^{*}(\hat{\beta}_{t}),\gamma_{0}\right\} -n^{-1/2}\sum_{i=1}^{n}\varphi_{a}\left\{ Z_{i}^{*}(\beta_{t}),\gamma_{0}\right\} \Big]+o_{\pr}(1).\label{eq:mean_imp}
\end{align}
The first term in formula \eqref{eq:mean_imp} is the sum of i.i.d.
components with $\E\left[\varphi_{a}\left\{ Z_{i}^{*}(\beta_{t}),\gamma_{0}\right\} \right]=0$.
Then by the central limit theorem, the first term converges to a normal
distribution with the mean $0$ and the variance $\V\left[\varphi_{a}\left\{ Z_{i}^{*}(\beta_{t}),\gamma_{0}\right\} \right]$.

For the term $n^{-1/2}\sum_{i=1}^{n}\varphi_{a}\left\{ Z_{i}^{*}(\hat{\beta}_{t}),\gamma_{0}\right\} -n^{-1/2}\sum_{i=1}^{n}\varphi_{a}\left\{ Z_{i}^{*}(\beta_{t}),\gamma_{0}\right\} $
in formula \eqref{eq:mean_imp}, consider a Taylor expansion of $n^{-1}\sum_{i=1}^{n}\varphi_{a}\left\{ Z_{i}^{*}(\hat{\beta}_{t}),\gamma_{0}\right\} $
with respect to $\hat{\beta_{t}}$ around $\beta_{t}$, again by the
regularity conditions C13 and C14, we have
\begin{align*}
\frac{1}{n}\sum_{i=1}^{n}\varphi_{a}\left\{ Z_{i}^{*}(\hat{\beta}_{t}),\gamma_{0}\right\}  & =\frac{1}{n}\sum_{i=1}^{n}\varphi_{a}\left\{ Z_{i}^{*}(\beta_{t}),\gamma_{0}\right\} +\frac{1}{n}\sum_{i=1}^{n}\frac{\partial\varphi_{a}\left\{ Z_{i}^{*}(\beta_{t}),\gamma_{0}\right\} }{\partial\beta_{t}^{\T}}(\hat{\beta}_{t}-\beta_{t})+o_{\pr}(n^{-1/2})\\
 & =\frac{1}{n}\sum_{i=1}^{n}\varphi_{a}\left\{ Z_{i}^{*}(\beta_{t}),\gamma_{0}\right\} +\frac{\partial\E\left[\varphi_{a}\left\{ Z_{i}^{*}(\beta_{t}),\gamma_{0}\right\} \right]}{\partial\beta_{t}^{\T}}(\hat{\beta}_{t}-\beta_{t})+o_{\pr}(n^{-1/2}).
\end{align*}

Note that 
\[
\frac{\partial\E\left[\varphi_{a}\left\{ Z_{i}^{*}(\beta_{t}),\gamma_{0}\right\} \right]}{\partial\beta_{t}^{\T}}=\left(\frac{\partial\E\left[\varphi_{a}\left\{ Z_{i}^{*}(\beta_{t}),\gamma_{0}\right\} \right]}{\partial\beta_{t,0}^{\T}},\cdots,\frac{\partial\E\left[\varphi_{a}\left\{ Z_{i}^{*}(\beta_{t}),\gamma_{0}\right\} \right]}{\partial\beta_{t,t-1}^{\T}}\right).
\]
From formula \eqref{eq:impute_value}, each component of the derivative
$\partial\E\left[\varphi_{a}\left\{ Z_{i}^{*}(\beta_{t}),\gamma_{0}\right\} \right]/\partial\beta_{t}^{\T}$
can be obtained by the Chain Rule as
\begin{align*}
\frac{\partial\E\left[\varphi_{a}\left\{ Z_{i}^{*}(\beta_{t}),\gamma_{0}\right\} \right]}{\partial\beta_{t,s-1}^{\T}} & =\E\left\{ R_{s-1}(1-R_{s})\frac{\partial\mu(A,X\mid\gamma_{0})}{\partial\gamma^{\T}}\frac{\partial\psi(e)}{\partial e}H_{s-1}^{\T}\right\} ,
\end{align*}
for $s=1,\cdots,t$. Then we can apply the linearization form stated
in Lemma \ref{lemma:beta_norm}, under the regularity condition C13
by Theorem 5.27 in \citet{boos2013essential}, we have $n^{-1}\sum_{i=1}^{n}\varphi_{a}\left\{ Z_{i}^{*}(\hat{\beta}_{t}),\gamma_{0}\right\} -n^{-1}\sum_{i=1}^{n}\varphi_{a}\left\{ Z_{i}^{*}(\beta_{t}),\gamma_{0}\right\} $
\[
=\frac{1}{n}\sum_{i=1}^{n}\sum_{s=1}^{t}\E\left\{ R_{s-1}(1-R_{s})\frac{\partial\mu(A,X\mid\gamma_{0})}{\partial\gamma^{\T}}\frac{\partial\psi(e)}{\partial e}H_{s-1}^{\T}\right\} U_{t,s-1,i}(\alpha_{0})+o_{\pr}(n^{-1/2}).
\]
Therefore, equation \eqref{eq:mean_imp} can be further expressed
as
\begin{align*}
\sqrt{n}(\hat{\gamma}-\gamma_{0}) & =n^{-1/2}\sum_{i=1}^{n}D_{\varphi}^{-1}\bigg[\varphi_{a}\left\{ Z_{i}^{*}(\beta_{t}),\gamma_{0}\right\} +\sum_{s=1}^{t}\E\left\{ R_{s-1}(1-R_{s})\frac{\partial\mu(A,X\mid\gamma_{0})}{\partial\gamma^{\T}}\frac{\partial\psi(e)}{\partial e}H_{s-1}^{\T}\right\} U_{t,s-1,i}(\alpha_{0})\bigg]\\
 & \qquad+o_{\pr}(1)\\
 & =n^{-1/2}\sum_{i=1}^{n}V_{\gamma,i}(\alpha_{0},\gamma_{0})+o_{\pr}(1).
\end{align*}

By the regularity condition C15, the ATE estimator $\hat{\tau}=n^{-1}\sum_{i=1}^{n}g(X_{i};\hat{\gamma}^{(0)})$
can be linearized as
\[
\hat{\tau}-\tau_{0}=\frac{1}{n}\sum_{i=1}^{n}\frac{\partial g(X_{i};\gamma_{0}^{(0)})}{\partial\gamma^{\T}}(\hat{\gamma}^{(0)}-\gamma_{0}^{(0)})+o_{\pr}(1).
\]
Since $\hat{\gamma}^{(0)}=c^{\T}\hat{\gamma}=(\mathbf{I}_{d_{0}},\mathbf{0}_{d_{0}\times d_{1}})\hat{\gamma}$,
by Theorem 1 and apply delta-method, we have the linearization
form of $\hat{\tau}$ as
\begin{align*}
\hat{\tau}-\tau_{0} & =\frac{1}{n}\sum_{i=1}^{n}\frac{\partial g(X_{i};\gamma_{0}^{(0)})}{\partial\gamma^{\T}}c^{\T}V_{\gamma,i}(\alpha_{0},\gamma_{0})+o_{\pr}(n^{-1/2})\\
 & =\frac{1}{n}\sum_{i=1}^{n}V_{\tau,i}(\alpha_{0},\gamma_{0})+o_{\pr}(n^{-1/2}).
\end{align*}
Under the regularity condition C16 and apply the central limit theorem,
we complete the proof. \end{proof}

\subsection{An example: using the interaction model for the ATE estimation \label{subsec:supp_interaction}}

The working model in the form of (2) in the main text covers a wide
range of analysis models in practice. We give an example of using
the interaction model for analysis, i.e., fit the regression model
with the interaction between the treatment variable and the baseline
covariates for the imputed data, as it is one of the most common models
in the clinical trials suggested in \citet{international2019addendum}.

\begin{example}\label{exmp:interaction} When using an interaction
model in the analysis step, the working model can be written as $\mu(A,X\mid\gamma)=AX^{\T}\gamma^{(0)}-X^{\T}\gamma^{(1)}$,
and the ATE estimator $\hat{\tau}$ can then be obtained by solving
the estimating equations 
\begin{align*}
\sum_{i=1}^{n}\begin{pmatrix}\psi(Y_{it}^{*}-A_{i}X_{i}^{\T}\gamma^{(0)}-X_{i}^{\T}\gamma^{(1)})(A_{i}X_{i}^{\T},X_{i}^{\T})^{\T}\\
{\color{black}{\color{black}{\color{red}{\color{black}X_{i}^{\T}\gamma^{(0)}-\tau}}}}
\end{pmatrix} & =0.
\end{align*}
Denote $V_{i}=(A_{i}X_{i}^{\T},X_{i}^{\T})^{\T}$ and $\gamma_{0}=(\gamma_{0}^{(0)\T},\gamma_{0}^{(1)\T})^{\T}$
such that $\E\big\{\psi(Y_{it}^{*}-A_{i}X_{i}^{\T}\gamma^{(0)}-X_{i}^{\T}\gamma^{(1)})\allowbreak (A_{i}X_{i}^{\T},X_{i}^{\T})^{\T}\big\}=0$.
Applying Theorems 1 and 2, the estimator
$\hat{\tau}\xrightarrow{\pr}\tau_{0}$ and $\sqrt{n}(\hat{\tau}-\tau_{0})\xrightarrow{d}\mathcal{N}\Big(0,\mathbb{V}\left\{ V_{\tau,i}(\alpha_{0},\gamma_{0},\mu_{X})\right\} \Big),$
where $V_{\tau,i}(\alpha_{0},\gamma_{0},\mu_{X})=(X_{i}-\mu_{X})^{\T}\gamma_{0}^{(0)}+\mu_{X}^{\T}V_{\gamma^{(0)},i}(\alpha_{0},\gamma_{0})$,
\begin{align*}
V_{\gamma^{(0)},i}(\alpha_{0},\gamma_{0}) & =c^{\T}D_{\varphi}^{-1}\bigg[\psi(e_{i})V_{i}+\sum_{s=1}^{t}\E\left\{ R_{is-1}(1-R_{is})V_{i}\frac{\partial\psi(e_{i})}{\partial e_{i}}H_{is-1}^{\T}\right\} U_{t,s-1,i}(\alpha_{0})\bigg],
\end{align*}
$U_{t,s-1,i}(\alpha_{0})=\left(\I_{p+s-2},\alpha_{s-1,0}\right)U_{t,s,i}(\alpha_{0})+\left(\mathbf{0}_{p+s-2}^{\T},1\right)\beta_{t,s}q(H_{is},\alpha_{s-1,0})$,
$U_{t,t-1,i}(\alpha_{0})=q(H_{it},\alpha_{t-1,0})$, and 
\[
q(H_{is},\alpha_{s-1,0})=\left[-\frac{\partial\E\left\{ \varphi(H_{is},\alpha_{s-1,0})H_{is-1}^{\T}\mid H_{is-1}\right\} }{\partial\alpha_{s-1}^{\T}}\right]^{-1}\varphi(H_{is},\alpha_{s-1,0}).
\]
Here, $c=(\I_{p},\mathbf{0}_{p\times p})$, where $\I_{p}$ is a $(p\times p)$-dimensional
identity matrix, $\mathbf{0}_{p\times p}$ is a $(p\times p)$-dimensional
zero matrix, $e_{i}=Y_{it}^{*}(\beta_{t})-V_{i}^{\T}\gamma^{*}$,
$D_{\varphi}=\partial\E\left\{ \psi(e_{i})V_{i}\right\} /\partial\gamma^{\T}$,
and 
\[
\begin{cases}
\beta_{t,t-1}=\alpha_{t-1,0} & \text{if \ensuremath{s=t}},\\
\beta_{t,s-1}=(\I_{p+s-2},\alpha_{s-1,0})(\I_{p+s-1},\alpha_{s,0})\cdots(\I_{p+t-3},\alpha_{t-2,0})\alpha_{t-1,0} & \text{if \ensuremath{s<t}},
\end{cases}
\]
for $s=1,\cdots,t$.

\end{example}

The asymptotic variance in Example \ref{exmp:interaction} motivates
us to obtain a linearization-based variance estimator by plugging
in the estimated values as
\[
\hat{\V}(\hat{\tau})=\frac{1}{n^{2}}\sum_{i=1}^{n}\left\{ V_{\tau,i}(\hat{\alpha}^{w},\hat{\gamma},\hat{\mu}_{X})-\bar{V}_{\tau}(\hat{\alpha}^{w},\hat{\gamma},\hat{\mu}_{X})\right\} ^{2},
\]
where $\bar{V}_{\tau}(\hat{\alpha}^{w},\hat{\gamma},\hat{\mu}_{X})=n^{-1}\sum_{i=1}^{n}V_{\tau,i}(\hat{\alpha}^{w},\hat{\gamma},\hat{\mu}_{X})$,
$V_{\tau,i}(\hat{\alpha}^{w},\hat{\gamma},\hat{\mu}_{X})=(X_{i}-\hat{\mu}_{X})^{\T}\hat{\gamma}^{(0)}+\hat{\mu}_{X}^{\T}V_{\gamma^{(0)},i}(\hat{\alpha}^{w},\hat{\gamma})$,
\begin{align*}
V_{\gamma^{(0)},i}(\hat{\alpha}^{w},\hat{\gamma}) & =c^{\T}\hat{D}_{\varphi}^{-1}\psi(\hat{e}_{i})V_{i}+\sum_{s=1}^{t}\left\{ \frac{1}{n}\sum_{i=1}^{n}R_{is-1}(1-R_{is})V_{i}\frac{\partial\psi(\hat{e}_{i})}{\partial e_{i}}H_{is-1}^{\T}\right\} U_{t,s-1,i}(\hat{\alpha}^{w}),
\end{align*}
$U_{t,s-1,i}(\hat{\alpha}^{w})=\left(\I_{p+s-2},\hat{\alpha}_{s-1}^{w}\right)U_{t,s,i}(\hat{\alpha}^{w})+\left(\mathbf{0}_{p+s-2}^{\T},1\right)\hat{\beta}_{t,s}\hat{q}(H_{is},\hat{\alpha}_{s-1}^{w})$,
$U_{t,t-1,i}(\hat{\alpha}^{w})=\hat{q}(H_{it},\hat{\alpha}_{t-1}^{w})$,
and 
\[
\hat{q}(H_{is},\hat{\alpha}_{s-1}^{w})=\bigg(-\frac{1}{n}\sum_{i=1}^{n}\frac{\partial\varphi(H_{is},\hat{\alpha}_{s-1}^{w})}{\partial\alpha_{s-1}^{\T}}H_{is-1}^{\T}\bigg)^{-1}\varphi(H_{is},\hat{\alpha}_{s-1}^{w}).
\]
Also, $\hat{e}_{i}=Y_{it}^{*}-V_{i}^{\T}\hat{\gamma}$, $\hat{D}_{\varphi}=n^{-1}\sum_{i=1}^{n}\partial\psi(\hat{e}_{i})V_{i}/\partial\gamma^{\T}$,
and 
\[
\begin{cases}
\hat{\beta}_{t,t-1}=\hat{\alpha}_{t-1}^{w} & \text{if \ensuremath{s=t}},\\
\hat{\beta}_{t,s-1}=(\I_{p+s-2},\hat{\alpha}_{s-1}^{w})(\I_{p+s-1},\hat{\alpha}_{s}^{w})\cdots(\I_{p+t-3},\hat{\alpha}_{t-2}^{w})\hat{\alpha}_{t-1}^{w} & \text{if \ensuremath{s<t}},
\end{cases}
\]
for $s=1,\cdots,t$. In practice, $\hat{\mu}_{X}$ is estimated by
the overall mean of the baseline covariates. We can also use the nonparametric
bootstrap to obtain a replication-based variance estimator. In the
simulation studies and real data application, we use the interaction
model for analysis.

\section{Illustration of the sequential regression procedure \label{sec:supp_seqreg}}

\subsection{Sequential linear regression}

In the main text, the sequential linear regression is mentioned multiple
times in Sections 2, 3, and 4,
under the assumed scenario where the current outcomes and the the
historical covariates have a linear relationship. Since the imputation
model under J2R focuses on the control group, we fit the current observed
outcomes $Y_{s}$ in the control group against the historical information
$H_{s-1}$ via a linear model to get the model parameter estimator
$\hat{\alpha}_{s-1}$ by solving the estimating equations $\sum_{i=1}^{n}(1-A_{i})R_{is}H_{is-1}(Y_{is}-H_{is-1}^{\T}\alpha_{s-1})=0$.
Our proposed weighted sequential robust regression model, whose robust
loss function is of the form (1), is motivated by this
sequential linear regression model. 

\subsection{Extension to general sequential regression}

In Section 3 in the main text, we mentioned a possible extension to
the nonlinear relationship between the current outcomes and the historical
covariates for the ATE identification. We now provide some insights
into it. The key for the ATE identification under the PMM framework
is to form the assumption of the pattern-specific expectation $\E(Y_{it}\mid R_{is-1}=1,R_{is}=0,A_{i}=a)$,
which can be identified via the iterated expectations $\E(Y_{it}\mid H_{is-1},A_{i}=0)=\E\big\{\cdots\E(Y_{it}\mid H_{it-1},R_{it}=1,A_{i}=0)\cdots\mid H_{is-1},R_{is}=1,A_{i}=0\big\}$
based on Assumptions 1 and 2
under J2R. If a nonlinear relationship is suspected, we can consider
adding nonlinear terms in the parametric models or turn to flexible
models such as semiparametric models or machine learning models for
model fitting. One natural way to estimate the iterated expectation
$\E(Y_{it}\mid H_{is-1},A_{i}=0)$ is to fit the sequential regressions
(via flexible models or parametric models with nonlinear terms) in
backward order. We again focus on the control group and give the detailed
implementation steps as follows.

\begin{enumerate}
\setlength{\itemindent}{1.5em}

\item[\textbf{Step 1}.] For the participants who are fully observed,
i.e., with the observed indicator $R_{it}=1$, fit the regression
model on $Y_{it}$ against the history $H_{it-1}$. Use the fitted
model to predict the outcomes for those who are observed until $(t-1)$th
visit time. Denote the predicted outcomes as $\hat{\E}(Y_{it}\mid H_{it-1},R_{it}=1,A_{i}=0)$.

\item[\textbf{Step 2}.] For the participants who are observed until
$(t-1)$th visit, fit the regression model on the predicted outcomes
$\hat{\E}(Y_{it}\mid H_{it-1},R_{it}=1,A_{i}=0)$ obtained in Step
1 against the history $H_{it-2}$. Use the fitted model to predict
the outcomes for those who are observed until $(t-2)$th visit time.
Denote the predicted outcomes as $\hat{\E}\big\{ \hat{\E}(Y_{it}\mid H_{it-1},R_{it}=1,A_{i}=0)\mid H_{it-2},R_{it-1}=1,\allowbreak A_{i}=0\big\} $.

\item[\textbf{Step 3}.] Follow the similar procedure $(t-s-2)$ times
by fitting the regression model in backward order. Obtain the predicted
outcomes $\hat{\E}(Y_{it}\mid H_{is-1},A_{i}=0)$ at last.

\end{enumerate}

\section{Additional notes on the simulation studies \label{sec:supp_sim}}

\subsection{Simulation setting \label{subsec:supp_simuset}}

In the simulation studies, the sample size is 500 for each group.
The baseline covariates $X=(X_{1},X_{2})^{\T}$ are generated independently
by $X_{1}\sim\mathcal{N}(0,1)$ and $X_{2}\sim\text{Bernoulli}(0.3)$.
The longitudinal outcomes are generated sequentially: 
\begin{enumerate}
\item at $t=1$, generate $Y_{1}=0.5+X_{1}-0.2X_{2}+\varepsilon_{1}$ for
both groups;
\item at $t=2$, generate
\[
\begin{cases}
Y_{2}=0.4+0.14X_{1}+0.52X_{2}+0.01Y_{1}+\varepsilon_{2} & \text{ if \ensuremath{A=0};}\\
Y_{2}=1.79+0.35X_{1}-0.05X_{2}+0.33Y_{1}+\varepsilon_{2} & \text{ if \ensuremath{A=1};}
\end{cases}
\]
\item at $t=3$, generate
\[
\begin{cases}
Y_{3}=0.77+0.02X_{1}+0.06X_{2}+0.71Y_{1}+0.84Y_{2}+\varepsilon_{3} & \text{ if \ensuremath{A=0};}\\
Y_{3}=2.52+1.16X_{1}-0.51X_{2}-1.53Y_{1}+0.46Y_{2}+\varepsilon_{3} & \text{ if \ensuremath{A=1};}
\end{cases}
\]
\item at $t=4$, generate
\[
\begin{cases}
Y_{4}=1.44-0.45X_{1}-0.24X_{2}-0.50Y_{1}-0.39Y_{2}+0.53Y_{3}+\varepsilon_{4} & \text{\text{ if \ensuremath{A=0};}}\\
Y_{4}=2.72-0.46X_{1}-0.06X_{2}+0.91Y_{1}+0.19Y_{2}+0.70Y_{3}+\varepsilon_{4} & \text{\text{ if \ensuremath{A=1};}}
\end{cases}
\]
\item at $t=5$, generate
\[
\begin{cases}
Y_{5}=4.37-0.84X_{1}-0.31X_{2}+0.01Y_{1}+0.35Y_{2}-0.32Y_{3}+0.81Y_{4}+\varepsilon_{5} & \text{\text{ if \ensuremath{A=0};}}\\
Y_{5}=4.21-0.02X_{1}-1.26X_{2}+0.24Y_{1}-0.18Y_{2}+0.65Y_{3}+0.13Y_{4}+\varepsilon_{5} & \text{\text{ if \ensuremath{A=1};}}
\end{cases}
\]
\end{enumerate}
where $\varepsilon_{k}$ is from a distribution with mean $0$ and
standard deviation $\sigma_{k}$, and $\sigma=(\sigma_{1},\cdots,\sigma_{5})^{\T}=(2.0,1.8,2.0,2.1,2.2)^{\T}.$
For the missing mechanisms, We set $\phi_{11}=-3.5,\phi_{12}=-3.6,\phi_{21}=\phi_{22}=0.2$.
The tuning parameter in the weighted robust regression is set as 10.
We also try to use cross-validation to obtain the tuning parameters,
which leads to very similar results. Therefore, to save computation
time, the tuning parameter is fixed in the MC simulation as $q_{s-1}=10$
for $s=1,\cdots,5$.

We consider two cases with the existence of extreme outliers or a
heavy-tailed distribution as follows.
\begin{enumerate}
\item Data with/without extreme outliers: The error terms are generated
by $\varepsilon_{k}\sim\mathcal{N}(0,\sigma_{k}^{2})$ to form the
multivariate normal distribution (MVN). To create the outliers, we
randomly select 10 individuals from the 30 completers with the maximum
outcomes at the last visit point per group and multiply the original
values by three for all post-baseline outcomes. 
\item Data from a heavy-tailed distribution: We choose a common heavy-tailed
distribution as t distribution. The error terms are generated by $\varepsilon_{k}\sim(3/5)^{1/2}\sigma_{k}t_{5}$
to get the same variation as the normal distribution, where $t_{5}$
is the standard t-distribution with the degrees of freedom as 5. 
\end{enumerate}

\subsection{Additional simulation results \label{subsec:supp_simtab}}

For the data with/without extreme outliers, apart from Table 1(b) in the main text,
we consider two more cases to incorporate the outliers only in one
specific group, with the same approach to generate the outliers as
presented in the main text. We again compare all the methods in terms
of point and variance estimation, type-1 error, power, and RMSE. 

Similar to the interpretation from Table 1 in the
main text, Table \ref{table:extreme-add} validates the superiority
of the proposed robust method, as it shows unbiased point estimates,
well-controlled type-1 errors under $H_{0}$, and high powers under
$H_{1}$. 

\begin{table}[!htbp]
\centering{}\centering \caption{Simulation results under the normal distribution with extreme points
at all post-baseline visit points. Here the true value $\tau=71.18\%$.}
\label{table:extreme-add} \scalebox{1}{ \resizebox{\textwidth}{!}{%
\begin{tabular}{>{\raggedright}p{0.1\textwidth}>{\centering}p{0.1\textwidth}cccccccccccccc}
\hline 
 &  & Point est & True var & \multicolumn{2}{c}{Var est} &  & \multicolumn{2}{c}{Relative bias} &  & \multicolumn{2}{c}{Coverage rate} &  & \multicolumn{2}{c}{Power} & RMSE\tabularnewline
Case & \multicolumn{1}{c}{Method} & ($\times10^{-2}$) & ($\times10^{-2}$) & \multicolumn{2}{c}{($\times10^{-2}$)} &  & \multicolumn{2}{c}{($\%$)} &  & \multicolumn{2}{c}{($\%$)} &  & \multicolumn{2}{c}{($\%$)} & ($\times10^{-2}$)\tabularnewline
 &  &  &  & $\hat{V}_{1}$ & $\hat{V}_{\text{Boot}}$ &  & $\hat{V}_{1}$ & $\hat{V}_{\text{Boot}}$ &  & $\hat{V}_{1}$ & $\hat{V}_{\text{Boot}}$ &  & $\hat{V}_{1}$ & $\hat{V}_{\text{Boot}}$ & \tabularnewline
\hline 
\multirow{3}{0.1\textwidth}{Outliers only in control} & \multicolumn{1}{c}{MI} & 43.07 & 4.29 & 13.10 & 6.23 &  & 205.40 & 45.37 &  & 98.00 & 84.80 &  & 8.80 & 40.80 & 34.90\tabularnewline
 & LSE & 51.44  & 3.66  & 4.71  & 4.39  &  & 28.68  & 19.84  &  & 88.70  & 86.90  &  & 68.60  & 70.70  & 27.48 \tabularnewline
 & Robust & 74.86 & 3.44 & 3.54 & 3.48 &  & 2.77 & 1.03 &  & 94.80 & 93.90 &  & 98.10 & 97.80 & 18.91\tabularnewline
\hline 
\multirow{3}{0.1\textwidth}{Outliers only in treatment} & \multicolumn{1}{c}{MI} & 116.55 & 3.43 & 8.29 & 6.40 &  & 141.36 & 86.31 &  & 74.40 & 58.80 &  & 100.00 & 100.00 & 49.00\tabularnewline
 & LSE & 94.11  & 3.63  & 4.73  & 4.69  &  & 30.24  & 28.98  &  & 84.90  & 86.20  &  & 99.70  & 99.70  & 29.82 \tabularnewline
 & Robust & 67.71 & 3.28 & 3.41 & 3.38 &  & 3.94 & 3.04 &  & 94.50 & 93.70 &  & 95.40 & 96.00 & 18.44\tabularnewline
\hline 
\end{tabular}} }
\end{table}

We also conduct the simulations under $H_{0}$ for each case. Under
$H_{0}$, we choose the same sequential regression coefficients for
both the control group and the treatment group. In addition, the tuning
parameter in the missing mechanism model is set as $\phi_{11}=\phi_{12}=-3.5$
and $\phi_{21}=\phi_{22}=0.2$. To achieve the accuracy of $0.01$,
we choose the Monte Carlo sample size as $10,000$.

Table \ref{table:extreme-h0} presents the simulation results under
MVN and $H_{0}$ without or with extreme outliers. Although the point
estimates seem to be unbiased when outliers exist (since we generate
the outliers in the same way for both groups, the bias for each group
cancels off), the type-1 error is extremely far away from the empirical
value, suggesting huge variabilities for the MI and LSE methods. The
proposed robust method outperforms as we observe a well-controlled
type-1 error, satisfying point and variance estimation results. The first two rows of Figure 3 in the main text visualizes the simulation results.

\begin{table}[!htbp]
\centering{}\centering \caption{Simulation results under the normal distribution and $H_{0}$ without
or with extreme outliers. Here the true value $\tau=0$.}
\vspace{1ex}
\label{table:extreme-h0} \scalebox{1}{ \resizebox{\textwidth}{!}{%
\begin{tabular}{>{\raggedright}p{0.1\textwidth}>{\centering}p{0.1\textwidth}cccccccccccc}
\hline 
 &  & Point est & True var & \multicolumn{2}{c}{Var est} &  & \multicolumn{2}{c}{Relative bias} &  & \multicolumn{2}{c}{Type-1 error} &  & RMSE\tabularnewline
Case & \multicolumn{1}{c}{Method} & ($\times10^{-2}$) & ($\times10^{-2}$) & \multicolumn{2}{c}{($\times10^{-2}$)} &  & \multicolumn{2}{c}{($\%$)} &  & \multicolumn{2}{c}{($\%$)} &  & ($\times10^{-2}$)\tabularnewline
 &  &  &  & $\hat{V}_{1}$ & $\hat{V}_{\text{Boot}}$ &  & $\hat{V}_{1}$ & $\hat{V}_{\text{Boot}}$ &  & $\hat{V}_{1}$ & $\hat{V}_{\text{Boot}}$ &  & \tabularnewline
\hline 
\multirow{3}{0.1\textwidth}{No outliers} & \multicolumn{1}{c}{MI} & 0.02  & 2.76  & 3.73  & 2.75  &  & 35.35  & -0.23  &  & 2.12  & 5.17  &  & 16.60 \tabularnewline
 & LSE & 0.03  & 2.72  & 2.71  & 2.71  &  & -0.25  & -0.21  &  & 4.86  & 5.16  &  & 16.49 \tabularnewline
 & Robust & -0.94  & 2.94  & 2.89  & 2.96  &  & -1.50  & 0.73  &  & 4.96  & 5.06  &  & 17.17 \tabularnewline
\hline 
\multirow{3}{0.1\textwidth}{Outliers in both groups} & \multicolumn{1}{c}{MI} & 0.13  & 3.61  & 11.55  & 8.77  &  & 219.61  & 142.52  &  & 0.07  & 0.36  &  & 19.01 \tabularnewline
 & LSE & 0.01  & 3.90  & 5.49  & 5.53  &  & 40.92  & 41.93  &  & 1.89  & 2.17  &  & 19.74 \tabularnewline
 & Robust & -1.05  & 3.03  & 2.90  & 3.00  &  & -4.36  & -1.05  &  & 5.26  & 5.29  &  & 17.45 \tabularnewline
\hline 
\end{tabular}} }
\end{table}

Table \ref{table:mvt h0} presents the simulation results under MVT
and $H_{0}$. Although all the methods have unbiased point estimates,
the proposed robust method is more efficient as the MC variance and
RMSE are small. The last row of Figure 3 also visualizes the simulation
results.

\begin{table}[!htbp]
\centering \caption{Simulation results under the t-distribution and $H_{0}$. Here the
true value $\tau=0$.}
 \scalebox{1}{ \resizebox{\textwidth}{!}{%
\begin{tabular}{>{\centering}p{0.1\textwidth}cccccccccccc}
\toprule 
 & Point est & True var & \multicolumn{2}{c}{Var est} &  & \multicolumn{2}{c}{Relative bias} &  & \multicolumn{2}{c}{Type-1 error} &  & RMSE\tabularnewline
Method & ($\times10^{-2}$) & ($\times10^{-2}$) & \multicolumn{2}{c}{($\times10^{-2}$)} &  & \multicolumn{2}{c}{($\%$)} &  & \multicolumn{2}{c}{($\%$)} &  & ($\times10^{-2}$)\tabularnewline
 &  &  & $\hat{V}_{1}$ & $\hat{V}_{\text{Boot}}$ &  & $\hat{V}_{1}$ & $\hat{V}_{\text{Boot}}$ &  & $\hat{V}_{1}$ & $\hat{V}_{\text{Boot}}$ &  & \tabularnewline
\midrule
\multicolumn{1}{c}{MI} & -0.14 & 2.78 & 3.75 & 2.76 &  & 34.87 & -0.77 &  & 2.33 & 5.35 &  & 16.68\tabularnewline
LSE & -0.13 & 2.76 & 2.73 & 2.73 &  & -1.05 & -0.95 &  & 5.09 & 5.39 &  & 16.60\tabularnewline
Robust & -1.05 & 2.46 & 2.41 & 2.47 &  & -2.18 & 0.52 &  & 5.38 & 5.47 &  & 15.72\tabularnewline
\bottomrule
\end{tabular}} } \label{table:mvt h0}
\end{table}

\section{Additional notes on the real-data application \label{sec:supp_real}}

The repeated CD4 count data is available at \url{https://content.sph.harvard.edu/fitzmaur/ala/cd4.txt}.
It keeps track of the longitudinal CD4 counts during the first 40
weeks of the clinical trial. Since the original CD4 counts are highly
skewed, we conduct a log transformation to get the transformed CD4
count as $\log(\text{CD4}+1)$ and use it as the outcome of interest.
As the longitudinal outcomes are collected at 8-week intervals, we
factorize the continuous-time variable into the intervals $(0,12]$,
$(12,20]$, $(20,28]$, $(28,36]$ and $(36,40]$. To ensure that
only one outcome is involved in a time interval for each individual,
only the outcome that is nearest to week $8k$ in the $k$th visit
interval is preserved for $k=1,\cdots,5$. Since our proposed method
is only valid for a monotone missingness pattern, we delete the observations
after the first occurrence of missingness for each individual to create
a monotone missingness dataset and use it for further analysis. The
fully-observed baseline covariates consist of age, gender, and the
baseline log CD4 counts. The created data suffers from severe missingness.
In arm 1, only 34 participants complete the study, while 94 drop out
before week 12, 52 drop out before week 20, 47 drop out before week
28, 17 drop out before week 36, and 76 drop out before week 40; in
arm 2, only 46 participants complete the study, while 94 drop out
before week 12, 48 drop out before week 20, 51 drop out before week
28, 20 drop out before week 36, and 71 drop out before week 40. 

We first conduct a scrutiny of the data to check the existence of
extreme outliers and/or a violation of normality. Figure 1 in the main text
presents the spaghetti plots of the repeated CD4 counts separated
by each treatment. From the figure, there are no outstanding outliers
in the data. Arm 2 has a higher average of the CD4 counts than arm
1. 

Then we check for normality by fitting sequential linear regressions
on the current outcomes against all historical information in arm
1 and examining the conditional residuals at each visit point for
model diagnosis. Figure 2 in the main text presents the QQ normal
plots for the conditional residuals. Note that we only focus on the
data in arm 1 since the imputation model under J2R relies solely on
the data in the reference group. From the figure, heavier tails are
detected at each visit point beyond the confidence region. We further
conduct the Shapiro-Wilk normality test for the conditional residuals.
All the tests return p-values that are much smaller than $0.05$,
therefore we reject the null hypothesis and conclude that the data
does not follow a normal distribution. Moreover, we conduct a symmetry
test proposed by \citet{miao2006new} on the conditional residuals.
All the resulting p-values are larger than $0.05$ and suggests that
the residuals are symmetric around 0, which allows us to obtain valid
inferences of the ATE via the proposed robust methods. All the test
results are presented in Figure 2 at visit $s$ for $s=1,\cdots,5$.

In the implementation of the weighted robust method, the tuning parameters
in formula (1) are selected via cross-validation. Specifically,
to mitigate the impact of outliers that are existed in the covariates
in the imputation model, we first conduct the cross-validation to
select the tuning parameter at each visit point in the sequential
robust regression that returns the smallest MSE, then insert the chosen
tuning parameters in the imputation model and further select the tuning
parameter for the analysis model in each group by cross-validation.
The resulting tuning parameters for the imputation model are $(20,19.5,17.5,15,8)$.
The choice of tuning parameters is not sensitive to the final estimation.

\end{document}